\documentclass{article}%
\usepackage{amsfonts}
\usepackage{amsmath}
\usepackage{amssymb}
\usepackage{graphicx}%
\providecommand{\U}[1]{\protect\rule{.1in}{.1in}}
\newtheorem{theorem}{Theorem}

\newtheorem{lemma}[theorem]{Lemma}

\newtheorem{proposition}[theorem]{Proposition}
\newtheorem{remark}[theorem]{Remark}

\newenvironment{proof}[1][Proof]{\noindent\textbf{#1.} }{\ \rule{0.5em}{0.5em}}
\begin{document}

\title{On monotone metric of classical channel and distribution spaces: asymptotic theory}
\author{Keiji Matsumoto}
\maketitle

\begin{abstract}
The aim of the manuscript is to characterize monotone metric in the space of
Markov map. Here, metric is not necessarily Riemanian, i.e., may not be the
inner product of the vector with itself.

So far, there have been plenty of literatures on the metric in the space of
probability distributions and quantum states. Among them, Cencov and
Petz\thinspace\ characterized all the monotone metrics in the classical and
quantum state space. \ As for channels, however, only a little is known about
its geometrical structures.

In that author's previous manuscript , the upper and the lower bound of
monotone channel metric was derived using resource conversion theory, and it
is proved that any monotone metric cannot be Riemanian.

Due to the latter result, we cannot rely on Cencov's theory, to build a
geometric theory consistent across probability distributions and channels. To
dispense with the assumption that a metric is Riemanian, we introduce some
assumptions on asymptotic behavior, \textit{weak asymptotic additivity} and
\textit{lower asymptotic continuity}. The proof utilizes resource conversion
technique. In the end of the paper, an implication on quantum state metrics is discussed.

\end{abstract}

\section{Introduction}

The aim of the manuscript is to characterize monotone metric in the space of
Markov map. Here, metric means the square of the norm defined on the tangent
space, and not necessarily Riemanian, nor induced from an inner product.

So far, there have been plenty of literatures on the metric in the space of
probability distributions and quantum states. Cencov, sometime in 1970s,
proved the monotone Riemanian metric in probability distribution space is
unique up to constant multiple, and identical to Fisher information metric
\cite{Cencov}. He also discussed invariant connections in the same space.
Amari and others independently worked on the same objects, especially from
differential geometrical view points, and applied to number of problems in
mathematical statistics, learning theory, time series analysis, dynamical
systems, control theory, and so on\cite{Amari}\cite{AmariNagaoka}. Quantum
mechanical states are discussed in literatures such as \cite{AmariNagaoka}%
\cite{Fujiwara}\cite{Matsumoto}\cite{Matsumoto}\cite{Petz}. Among them,
Petz\thinspace\cite{Petz} characterized all the monotone Riemanian metrics in
the quantum state space using operator mean theory.

As for channels, however, much less is known. To my knowledge, there had been
no study about axiomatic characterization of distance measures in the
classical or quantum channel space, except for the author's
manuscript\thinspace\cite{Matsumoto:2010-1}. In that manuscript, the upper and
the lower bound of monotone channel metric was derived using resource
conversion theory, and it is proved that any monotone metric cannot be Riemanian.

The latter result has some impact on the axiomatic theory of the monotone
metric in the space of classical and quantum states, since both
Cencov\thinspace\cite{Cencov} and Petz\thinspace\cite{Petz} assumed metrics
are Riemanian. Since classical and quantum states can be viewed as channels
with the constant output, it is preferable to dispense with this assumption.
Recalling that the Fisher information is useful in asymptotic theory, it would
be natural to introduce some assumptions on their asymptotic behavior. Hence,
we introduced \textit{weak asymptotic additivity} and \textit{lower asymptotic
continuity}. By these additional assumptions, we not only recovers uniqueness
result of Cencov \thinspace\cite{Cencov}, but also proves uniqueness of the
monotone metric in the channel space.

In this proof, again, we used resource conversion technique. A difference from
usual resource conversion technique is that asymptotic continuity is replaced
by a bit weaker \textit{lower asymptotic continuity.} The reason is that the
former condition is not satisfied by Fisher information.

In the end, there is an implication on quantum state metrics.

\section{Notations and conventions}

In discussing probability distributions, the underlying set is denoted by
$\mathcal{\Omega}$. In discussing channels, $\mathcal{\Omega}_{\mathrm{in}}$
($\mathcal{\Omega}_{\mathrm{out}}$) denotes the totality of the inputs
(outputs). In this paper, they are either $\left\{  1,\cdots,k\right\}  $ or $%
\mathbb{R}
^{d}$. $x$,$y$, etc. denotes an element of $\mathcal{\Omega}_{\mathrm{in}}$
,$\mathcal{\Omega}_{\mathrm{out}}$, $\mathcal{\Omega}$. Also, $x^{n}=\left(
x_{1},x_{2},\cdots,x_{n}\right)  $, $y^{n}=\left(  y_{1},y_{2},\cdots
,y_{n}\right)  $, etc. denotes an element of $\mathcal{\Omega}^{\times n}$,
$\mathcal{\Omega}_{\mathrm{in}}^{\times n}$ ,$\mathcal{\Omega}_{\mathrm{out}%
}^{\times n}$ .

Random variable taking values in $\mathcal{\Omega}$, $\mathcal{\Omega
}_{\mathrm{in}}$ ,$\mathcal{\Omega}_{\mathrm{out}}$ are denoted by $X$, $Y$,
while random variable taking values in $\mathcal{\Omega}^{\times n}$,
$\mathcal{\Omega}_{\mathrm{in}}^{\times n}$ ,$\mathcal{\Omega}_{\mathrm{out}%
}^{\times n}$ are denoted by $X^{n}$, $Y^{n}$. The dsitribution of $X$ is
denoted by $\mathbb{P}_{X}$, while its density (with respect to lebesgue
measure or counting measure depending on the underlying set) is denoted by
$\mathrm{p}_{X}$. $\ \mathbb{P}_{X|Y}$ and $\mathrm{p}_{X|Y}$ denotes the
onditional distribution and its density, respectively. In this paper, the
existence of density with respect to a standard underlying measure $\mu$
(counting measure for $\left\{  1,\cdots,k\right\}  $, and Lebesgue measure
for $%
\mathbb{R}
^{d}$ ) is always assumed. Hence, by abusing the term, we sometimes say
`distribution $p$'. By $\mathcal{P}$, $\mathcal{P}_{\mathrm{in}}$, and
$\mathcal{P}_{\mathrm{out}}$ we denote the totality of the probability density
functions over $\mathcal{\Omega}$, $\mathcal{\Omega}_{\mathrm{in}}$, and
$\mathcal{\Omega}_{\mathrm{out}}$, respectively.

Channel $\Phi$ is a linear map from probability distributions over to
$\mathcal{\Omega}_{\mathrm{in}}$ to those over $\mathcal{\Omega}%
_{\mathrm{out}}$, but also considered as a map from $L_{1}\left(
\mathcal{\Omega}_{\mathrm{in}}\right)  $ to $L_{1}\left(  \mathcal{\Omega
}_{\mathrm{out}}\right)  $. Hence, we use notation such as $\Phi\left(
\mathbb{P}_{X}\right)  $ , as well as $\Phi\left(  \mathrm{p}_{X}\right)  $.
The totality of channels is denoted by $\mathcal{C}$. If there is a need to
indicate input and output space, we use the notation such as $\mathcal{C}%
\left(  \mathcal{P}_{\mathrm{in}},\mathcal{P}_{\mathrm{out}}\right)  $.
$\Phi^{\ast}$ denotes the dual map of $\Phi$,%
\[
\int f\left(  x\right)  \Phi\left(  \mathrm{p}_{X}\right)  \left(  x\right)
\mathrm{d}\mu\left(  x\right)  =\int\Phi^{\ast}\left(  f\right)  \left(
x\right)  \mathrm{p}_{X}\left(  x\right)  \mathrm{d}\mu\left(  x\right)  .
\]

A tangent space is denoted by a notation $\mathcal{T}_{\cdot}\left(
\cdot\right)  $. $\ \delta$, $\delta^{\prime}$ etc. denotes an element of
$\mathcal{T}_{p}\left(  \mathcal{P}\right)  $ (the tangent space to the set
$\mathcal{P}$ at the point $p$) etc, w hile $\Delta$, $\Delta^{\prime}$ etc
denotes an element of $\mathcal{T}_{\Phi}\left(  \mathcal{C}\right)  $ etc.
\ In the paper, we identify $\delta\in\mathcal{T}_{p}\left(  \mathcal{P}%
\right)  $ with an element of $L^{1}$ in the form of $c\left(  p_{1}%
-p_{2}\right)  $, where $p_{1}$, $p_{2}\in\mathcal{P}$ . Hence, the
differential map of $\Phi$ is also denoted by $\Phi$, by abusing the notation.
$L$ is a random variable defined by
\[
L\left(  x\right)  =\frac{\delta\left(  x\right)  }{p\left(  x\right)
},\text{\thinspace\ }x\in\Omega.\text{\thinspace\thinspace}%
\]
and its low is under $p$, unless otherwise mentioned. Also, $\Delta$ is
identified with a linear map in the form of $c\left(  \Psi_{1}-\Psi
_{2}\right)  $, where $\Psi_{1}$, $\Psi_{2}\in\mathcal{C}$.

A pair $\left\{  p,\delta\right\}  $ and $\left\{  \Phi,\Delta\right\}  $ is
called \textit{local data at }$p$ and $\Phi$, respectively, since it decides
local behaviour of one-parameter family of distributions at the point $p$ and
$\Phi$, respectively. We denote by $\mathrm{N}\left(  a,\sigma^{2}\right)  $
and $\delta\mathrm{N}\left(  a,\sigma^{2}\right)  $ the Gaussian distribution
with mean $a$ and variance $\sigma^{2}$ and singed measure defined by
\[
\delta\mathrm{N}\left(  a,\sigma^{2}\right)  \left(  B\right)  :=\frac
{1}{\sqrt{2\pi}\sigma}\int_{B}\frac{x-a}{\sigma^{2}}\exp\left[  -\frac
{1}{2\sigma^{2}}\left(  x-a\right)  ^{2}\right]  \mathrm{d}x,
\]
respectively. Thus, the local data $\left\{  \mathrm{N}\left(  a,\sigma
^{2}\right)  ,\delta\mathrm{N}\left(  a,\sigma^{2}\right)  \right\}  $
describes local behaviour of Gaussian shift family $\left\{  \mathrm{N}\left(
\theta,\sigma^{2}\right)  \right\}  _{\theta\in%
\mathbb{R}
}$ at $\theta=a$.

The symbol `$\otimes$' means direct product of vectors. Given $f\in
L_{1}\left(  \mathcal{\Omega}_{1}\right)  $ and $g\in L_{1}\left(
\mathcal{\Omega}_{2}\right)  $, $f\otimes g$ is defined by
\[
f\otimes g\left(  x_{1},x_{2}\right)  =f\left(  x_{1}\right)  f\left(
x_{2}\right)  .
\]
The linear span of $\left\{  f\otimes g\right\}  $ is denoted by $L_{1}\left(
\mathcal{\Omega}_{1}\right)  \otimes L_{2}\left(  \mathcal{\Omega}_{2}\right)
$($=L_{1}\left(  \mathcal{\Omega}_{1}\times\Omega_{2}\right)  $). Also, given
$\Phi_{1}\in\mathcal{C}\left(  \mathcal{P}_{\mathrm{in,1}},\mathcal{P}%
_{\mathrm{out,1}}\right)  $, $\Phi_{2}\in\mathcal{C}\left(  \mathcal{P}%
_{\mathrm{in,2}},\mathcal{P}_{\mathrm{out,2}}\right)  $, $\Phi_{1}\otimes
\Phi_{2}\in\mathcal{C}\left(  \mathcal{P}_{\mathrm{in,1}}\otimes
\mathcal{P}_{\mathrm{in,2}},\mathcal{P}_{\mathrm{out,1}}\otimes\mathcal{P}%
_{\mathrm{out,2}}\right)  $ is defined by the relation
\[
\Phi_{1}\otimes\Phi_{2}\left(  f\otimes g\right)  =\Phi_{1}\left(  f\right)
\otimes\Phi_{2}\left(  g\right)
\]
and linearity. For a real valued random variable $F_{1}$ and $F_{2}$ over
$\Omega_{1}$ and $\Omega_{2}$, respectively, $F_{1}\otimes F_{2}$ is a random
variable over $\Omega_{1}\times\Omega_{2}$ with
\[
F_{1}\otimes F_{2}\left(  x_{1},x_{2}\right)  =F\left(  x_{1}\right)  F\left(
x_{2}\right)  .
\]

We use abbraviations such as $f^{\otimes n}:=f\otimes f\otimes\cdots\otimes
f$, and%
\begin{align*}
\delta^{\left(  n\right)  }  &  :=\delta\otimes p^{\otimes n-1}+p\otimes
\delta\otimes p^{\otimes n-2}+\cdots+p^{\otimes n-1}\otimes\delta
\in\mathcal{T}_{p}\left(  \mathcal{P}^{\otimes n}\right)  ,\\
L^{\left(  n\right)  }  &  :=L\otimes1^{\otimes n-1}+1\otimes L\otimes
1^{\otimes n-2}+\cdots+1^{\otimes n-1}\otimes L,\\
\Delta^{\left(  n\right)  }  &  :=\Delta\otimes\Phi^{\otimes n-1}+\Phi
\otimes\Delta\otimes\Phi^{\otimes n-2}+\cdots+\Phi^{\otimes n-1}\otimes
\Delta\in\mathcal{T}_{\Phi}\left(  \mathcal{C}^{\otimes n}\right)  ,\\
\left\{  p,\delta\right\}  ^{\otimes n}  &  :=\left\{  p^{\otimes n}%
,\delta^{\left(  n\right)  }\right\}  ,\\
\left\{  \Phi,\Delta\right\}  ^{\otimes n}  &  :=\left\{  \Phi^{\otimes
n},\Delta^{\left(  n\right)  }\right\}  ,\\
\left\{  p_{1},\delta_{1}\right\}  \otimes\left\{  p_{2},\delta_{2}\right\}
&  :=\left\{  p_{1}\otimes p_{2},\delta_{1}\otimes p_{2}+p_{1}\otimes
\delta_{2}\right\}  .
\end{align*}
$\ \ \ \left\Vert \cdot\right\Vert _{1}$ denotes, for a (singed) measure,
total variation, and for a function, $L_{1}$-norm. $\left\Vert \cdot
\right\Vert _{\mathrm{cb}}$ denotes completely bounded norm: for a linear map
$\Lambda$ form signed measures ($L^{1}$-functions ) to signed measures
($L^{1}$-functions),
\[
\left\Vert \Lambda\right\Vert _{\mathrm{cb}}=\max_{p\text{: probability
distributions}}\left\Vert \Lambda\otimes\mathbf{I}\left(  p\right)
\right\Vert _{1}.
\]
(Here note $\Lambda$ may not be a Markov map, i.e., may not map a probability
distirbution to another dsitribution.)

$g_{p}\left(  \delta\right)  $ and $G_{\Phi}\left(  \Delta\right)  $ denotes
\ a metric, or square of a norm in $\mathcal{T}_{p}\left(  \mathcal{P}\right)
$ and $\mathcal{T}_{\Phi}\left(  \mathcal{C}\right)  $, respectively. In the
present paper, they are not necessarily Riemanian. A probability distribution
$p$ is identified with the Markov map which sends all the input probability
distributions to $p$, so that notations such as $G_{p}\left(  \delta\right)  $
makes sense. $J_{p}\left(  \delta\right)  $ denotes Fisher information,%
\[
J_{p}\left(  \delta\right)  :=\mathbb{E}\left\{  L\right\}  ^{2}=\int\left\{
L\left(  x\right)  \right\}  ^{2}p\left(  x\right)  \mathrm{d}\mu\left(
x\right)  =\int\frac{\left\{  \delta\left(  x\right)  \right\}  ^{2}}{p\left(
x\right)  }\mathrm{d}\mu\left(  x\right)
\]
$\ \ $\ $\ $Finally, $\Phi\left(  \cdot|x\right)  \in\mathcal{P}%
_{\mathrm{out}}$ is the distribution (, or its density) of the output when the
input is $x$. Also, with $\Delta=c\left(  \Phi_{1}-\Phi_{2}\right)  ,$%

\[
\Delta\left(  \cdot|x\right)  :=c\left(  \Phi_{1}\left(  \cdot|x\right)
-\Phi_{2}\left(  \cdot|x\right)  \right)  \in\mathcal{T}_{p}\left(
\mathcal{P}_{\mathrm{out}}\right)  .
\]

\section{Probability distributions}

Cencov had proven uniqueness (up to the constant multiple) of the monotone
metric in the space of classical probability distributions defined over the
finite set. In the proof, it is essential that the metric is Riemanian, i.e.,
induced from an inner product. As will be noted in Theorem\thinspace
\ref{th:no-inner-product}, however, this assumption is not compatible with
monotonicity in case of channels. Hence, we dispense with this assumption,
and, instead, introduce new axioms which rules asymptotic behaviour of a metric.

\subsection{Axioms for the metrics of probability distributions}

\begin{description}
\item[(M0)] $g_{p}\left(  \delta\right)  \geq g_{\Psi\left(  p\right)
}\left(  \Psi\left(  \delta\right)  \right)  $.

\item[(A0)] $\lim_{n\rightarrow\infty}\frac{1}{n}$ $g_{p^{\otimes n}}\left(
\delta^{\left(  n\right)  }\right)  =g_{p}\left(  \delta\right)  $.

\item[(C0)] If $\left\Vert q^{n}-p^{\otimes n}\right\Vert _{1}\rightarrow0$
and $\frac{1}{\sqrt{n}}\left\Vert \delta^{\prime n}-\delta^{\left(  n\right)
}\right\Vert _{1}\rightarrow0$ then
\[
\varliminf_{n\rightarrow\infty}\frac{1}{n}\left(  g_{q^{n}}\left(
\delta^{\prime n}\right)  -g_{p^{\otimes n}}\left(  \delta^{\left(  n\right)
}\right)  \right)  \geq0.
\]

\item[(N0)] (Normalization) In case $\left\{  p,\delta\right\}  =\left\{
\mathrm{N}\left(  0,1\right)  ,\delta\mathrm{N}\left(  0,1\right)  \right\}  $
(a Gaussian shift family),
\[
g_{p}\left(  \delta\right)  =1.
\]

\end{description}

\subsection{Simulation and asymptotic tangent smulation:definition}

\textit{Simulation} of $\left\{  p_{\theta}\right\}  $ is the pair $\left\{
q_{\theta},\Lambda\right\}  $ with%
\[
p_{\theta}=\Lambda\left(  q_{\theta}\right)  ,\,\,\forall\theta\in\Theta,
\]
and \textit{tangent simulation} of the local data $\left\{  p,\delta\right\}
$ is the pair $\left\{  q,\delta^{\prime},\Lambda\right\}  $ with
\[
p=\Lambda\left(  q\right)  ,\,\,\,\delta=\Lambda\left(  \delta^{\prime
}\right)  .
\]
If in addition there is $\Lambda^{\prime}$ with
\[
q=\Lambda^{\prime}\left(  p\right)  ,\,\,\delta^{\prime}=\Lambda^{\prime
}\left(  \delta\right)  ,
\]
we say $\left\{  p,\delta\right\}  $ and $\left\{  q,\delta^{\prime}\right\}
$ are \textit{equivalent}, and express this relation by the notaton
\[
\left\{  p,\delta\right\}  \equiv\left\{  q,\delta^{\prime}\right\}  .
\]

An \textit{asymptotic tangent simulation} of $\left\{  p^{\otimes n}%
,\delta^{\left(  n\right)  }\right\}  $ means a sequence $\left\{
q^{n},\delta^{\prime n},\Lambda^{n}\right\}  _{n=1}^{\infty}$ of triplet of a
probability density $q^{n}$, an $L^{1}$-function $\delta^{\prime n}$ with
$\int\delta^{\prime n}\mathrm{d}\mu=0$, and \ a Markov map $\Lambda^{n}$, such
that%
\begin{align}
\lim_{n\rightarrow\infty}\left\Vert p^{\otimes n}-\Lambda^{n}\left(
q^{n}\right)  \right\Vert _{1}  &  =0,\\
\lim_{n\rightarrow\infty}\frac{1}{\sqrt{n}}\left\Vert \delta^{\left(
n\right)  }-\Lambda^{n}\left(  \delta^{\prime n}\right)  \right\Vert _{1}  &
=0.
\end{align}
\ We call $\max\left\{  \left\Vert p^{\otimes n}-\Lambda^{n}\left(
q^{n}\right)  \right\Vert _{1},\left\Vert p^{\otimes n}-\Lambda^{n}\left(
q^{n}\right)  \right\Vert _{1}\right\}  $ the \textit{error} of the asymptotic
tangent simulation. In all the cases treated int the present paper, the
following stronger conditions are satisfied:%
\begin{align}
\left\Vert p^{\otimes n}-\Lambda^{n}\left(  q^{n}\right)  \right\Vert _{1}  &
\leq\frac{1}{\sqrt{n}}C\left(  \left\{  p,\delta\right\}  \right)
,\label{simulation-tangent-prob-1}\\
\frac{1}{\sqrt{n}}\left\Vert \delta^{\left(  n\right)  }-\Lambda^{n}\left(
\delta^{\prime n}\right)  \right\Vert _{1}  &  \leq\frac{1}{n^{1/4}}C\left(
\left\{  p,\delta\right\}  \right)  . \label{simulation-tangent-prob-2}%
\end{align}
Below, $C\left(  \left\{  p,\delta\right\}  \right)  $ is sometimes denoted by
$C$, as long as no confusion is likely to arise.

\begin{proposition}
\label{prop:sufficient}Let $L\left(  x\right)  :=\delta\left(  x\right)
/p\left(  x\right)  $. Then,
\[
\left\{  p,\delta\right\}  \equiv\left\{  \mathrm{p}_{L}\left(  l\right)
,\,l\mathrm{p}_{L}\left(  l\right)  \right\}  .
\]

\end{proposition}

\begin{proof}
Observe
\begin{align*}
\int_{x:L\left(  x\right)  =l}p\left(  x\right)  \mathrm{d}\mu\left(
x\right)   &  =\mathrm{p}_{L}\left(  l\right)  ,\\
\int_{x:L\left(  x\right)  =l}\delta\left(  x\right)  \mathrm{d}\mu\left(
x\right)   &  =\int_{x:L\left(  x\right)  =l}L\left(  x\right)  p\left(
x\right)  \mathrm{d}\mu\left(  x\right) \\
&  =l\int_{x:L\left(  x\right)  =l}p\left(  x\right)  \mathrm{d}\mu\left(
x\right)  =l\mathrm{p}_{L}\left(  l\right)  ,
\end{align*}
where $\mu$ is either Lebesgue measure ($\Omega=%
\mathbb{R}
^{d}$) or counting measure ($\Omega=\{1,\cdots,k\}$). Also,
\begin{align*}
\mathrm{p}_{X|L}\left(  x|l\right)  \mathrm{p}_{L}\left(  l\right)   &
=\left\{
\begin{array}
[c]{cc}%
p\left(  x\right)  , & \left(  l=L\left(  x\right)  \right) \\
0, & \text{otherwise}%
\end{array}
\right.  ,\\
\mathrm{p}_{X|L}\left(  x|l\right)  \,\left\{  l\mathrm{p}_{L}\left(
l\right)  \right\}   &  =\left\{
\begin{array}
[c]{cc}%
L\left(  x\right)  p\left(  x\right)  =\delta\left(  x\right)  , & \left(
l=L\left(  x\right)  \right) \\
0, & \text{otherwise}%
\end{array}
\right.  .
\end{align*}
Therefore, letting $\nu$ be a measure induced from $\mu$ via change of the
variable $l=\delta\left(  x\right)  /p\left(  x\right)  $,
\begin{align*}
\int\mathrm{p}_{X|L}\left(  x|l\right)  \mathrm{p}_{L}\left(  l\right)
\mathrm{d}\nu\left(  l\right)   &  =p\left(  x\right)  ,\\
\int\mathrm{p}_{X|L}\left(  x|l\right)  \left\{  l\mathrm{p}_{L}\left(
l\right)  \right\}  \mathrm{d}\nu\left(  l\right)   &  =\delta\left(
x\right)  .
\end{align*}

\end{proof}

\begin{lemma}
\label{lem:sim-tangent}Let $L^{^{\prime}n}:=\delta^{\prime n}/q^{n}$, and
suppose that $q^{n}=\mathrm{p}_{L^{\prime n}}$ . Let $\tilde{L}^{n}$ be a
random variable defined over $\mathcal{B}\left(  L^{\left(  n\right)
}\right)  $, obeying the distribution$\ $%
\begin{align*}
\mathrm{p}_{\tilde{L}^{n}}\left(  l^{n}\right)   &  :=\tilde{\Lambda}%
^{n}\left(  \mathrm{p}_{L^{\prime n}}\right)  \left(  l^{n}\right) \\
&  :=\int P^{n}\left(  l^{n}|l^{\prime n}\right)  \mathrm{p}_{L^{\prime n}%
}\left(  l^{\prime n}\right)  \mathrm{d}l^{\prime n}.
\end{align*}
\ Define $\Lambda^{n}$\textrm{ }by%
\[
\Lambda^{n}\left(  \mathrm{q}\right)  \left(  x^{n}\right)  :=\int
\tilde{\Lambda}^{n}\left(  \mathrm{q}\right)  \left(  l^{n}\right)
\mathrm{p}_{X^{n}|L^{\left(  n\right)  }}\left(  x^{n}|l^{n}\right)
\mathrm{d}l^{n}.\,
\]
For $\left\{  q^{n},\delta^{\prime n},\Lambda^{n}\right\}  _{n=1}^{\infty}$ to
satisfy (\ref{simulation-tangent-prob-1}) and (\ref{simulation-tangent-prob-2}%
), it sufficeas that \
\begin{equation}
\left\Vert \mathrm{p}_{L^{\left(  n\right)  }}-\mathrm{p}_{\tilde{L}^{n}%
}\right\Vert _{1}\leq\frac{C^{\prime}}{\sqrt{n}},
\label{simulation-tangent-prob-1-2}%
\end{equation}
and
\begin{equation}
\max\left\{  \mathbb{E}\left\vert \mathbb{E}\left[  L^{\prime n}\,|\tilde
{L}^{n}\,\right]  -\tilde{L}^{n}\right\vert ,\,\,\mathbb{E}\left(  L\right)
^{2},\frac{1}{n}\mathbb{E}\left(  \tilde{L}^{n}\right)  ^{2}\right\}  \leq
a<\infty,\, \label{simulation-tangent-prob-2-2}%
\end{equation}
where
\[
2C^{\prime}+3a\leq C.
\]

\end{lemma}

\begin{proof}
Since $\mathbb{P}_{X^{n}|L^{\left(  n\right)  }}\left(  \mathcal{A}|l^{\left(
n\right)  }\right)  \leq1$, (\ref{simulation-tangent-prob-1-2}) implies
\[
\left\Vert p^{\otimes n}-\Lambda^{n}\left(  q^{n}\right)  \right\Vert
_{1}=\sup_{\mathcal{A}\text{:measurable}}\left\vert \mathbb{E\,P}%
_{X^{n}|L^{\left(  n\right)  }}\left(  \mathcal{A}|L^{\left(  n\right)
}\right)  -\mathbb{E\,P}_{X^{n}|L^{\left(  n\right)  }}\left(  \mathcal{A}%
|\tilde{L}^{n}\right)  \right\vert \leq\frac{C^{\prime}}{\sqrt{n}},
\]
which is (\ref{simulation-tangent-prob-1}). By Chebychev's inequality,
\begin{align*}
\mathbb{E}\left\{  \frac{1}{\sqrt{n}}\left\vert L^{\left(  n\right)
}\right\vert \,\,;\frac{1}{\sqrt{n}}\left\vert L^{\left(  n\right)
}\right\vert \,\geq n^{1/4}\right\}   &  \leq n^{-1/4}\cdot\frac{1}%
{n}\mathbb{E}\left(  L^{\left(  n\right)  }\right)  ^{2}\leq an^{-1/4},\\
\mathbb{E}\left\{  \frac{1}{\sqrt{n}}\left\vert \tilde{L}^{n}\right\vert
\,\,;\frac{1}{\sqrt{n}}\left\vert \tilde{L}^{n}\right\vert \,\geq
n^{1/4}\right\}   &  \leq an^{-1/4}.
\end{align*}
Also,
\begin{align*}
\tilde{\Lambda}^{n}\left(  \delta^{\prime n}\right)  \left(  l^{n}\right)   &
=\int l^{\prime n}P^{n}\left(  l^{n}|l^{\prime n}\right)  \mathrm{p}%
_{L^{\prime n}}\left(  l^{\prime n}\right)  \mathrm{d}l^{\prime n}\\
&  =\int l^{\prime n}\frac{P^{n}\left(  l^{n}|l^{\prime n}\right)
\mathrm{p}_{L^{\prime n}}\left(  l^{\prime n}\right)  }{\mathrm{p}_{\tilde
{L}^{n}}\left(  l^{n}\right)  }\mathrm{d}l^{\prime n}\,\mathrm{p}_{\tilde
{L}^{n}}\left(  l^{n}\right) \\
&  =\mathbb{E}\left[  L^{^{\prime}n}\,|\tilde{L}^{n}=l^{n}\,\right]
\mathrm{p}_{\tilde{L}^{n}}\left(  l^{n}\right)  .
\end{align*}
Therefore, by Proposition\thinspace\ref{prop:sufficient},
\begin{align*}
&  \frac{1}{\sqrt{n}}\left\Vert \delta^{\left(  n\right)  }-\Lambda^{n}\left(
\delta^{\prime n}\right)  \right\Vert _{1}\\
&  \leq\frac{1}{\sqrt{n}}\int\left\vert l^{n}\mathrm{p}_{L^{\left(  n\right)
}}\left(  l^{n}\right)  -\tilde{\Lambda}^{n}\left(  \delta^{\prime n}\right)
\left(  l^{n}\right)  \right\vert \mathrm{d}l^{n}\\
&  =\frac{1}{\sqrt{n}}\int\left\vert l^{n}\mathrm{p}_{L^{\left(  n\right)  }%
}\left(  l^{n}\right)  -\mathbb{E}\left[  L^{^{\prime}n}\,|\tilde{L}^{n}%
=l^{n}\,\right]  \mathrm{p}_{\tilde{L}^{n}}\left(  l^{n}\right)  \right\vert
\mathrm{d}l^{n}\\
&  \leq\frac{1}{\sqrt{n}}\int\left\vert l^{n}\mathrm{p}_{L^{\left(  n\right)
}}\left(  l^{n}\right)  -l^{n}\mathrm{p}_{\tilde{L}^{n}}\left(  l^{n}\right)
\right\vert \mathrm{d}l^{n}\\
&  +\frac{1}{\sqrt{n}}\int\left\vert l^{n}\mathrm{p}_{\tilde{L}^{n}}\left(
l^{n}\right)  -\mathbb{E}\left[  L^{^{\prime}n}\,|\tilde{L}^{n}=l^{n}%
\,\right]  \mathrm{p}_{\tilde{L}^{n}}\left(  l^{n}\right)  \right\vert
\mathrm{d}l^{n}\\
&  \leq\frac{1}{\sqrt{n}}\int\left\vert l^{n}\mathrm{p}_{L^{\left(  n\right)
}}\left(  l^{n}\right)  -l^{n}\mathrm{p}_{\tilde{L}^{n}}\left(  l^{n}\right)
\right\vert \mathrm{d}l^{n}+\frac{a}{n^{1/4}}\\
&  \leq\mathbb{E}\left\{  \frac{1}{\sqrt{n}}\mathbb{E}\left\vert L^{\left(
n\right)  }-\tilde{L}^{n}\right\vert ;\frac{1}{\sqrt{n}}\left\vert L^{\left(
n\right)  }\right\vert \,\leq n^{1/4},\frac{1}{\sqrt{n}}\left\vert \tilde
{L}^{n}\right\vert \,\leq n^{1/4}\right\}  +\frac{2a}{n^{1/4}}+\frac
{a}{n^{1/4}}\\
&  \leq2n^{1/4}\left\Vert p^{\otimes n}-\Lambda^{n}\left(  q^{n}\right)
\right\Vert _{1}+\frac{3a}{n^{1/4}}\\
&  \leq\frac{3a+2C^{\prime}}{n^{1/4}}\leq\frac{C}{n^{1/4}}.
\end{align*}

\end{proof}

\begin{proposition}
\label{prop:sym-chain}Suppose there is an asymptotic tangent simulation of
$\left\{  p_{i+1}^{n},\delta_{i+1}^{n}\right\}  $ by $\left\{  p_{i+1}%
^{n},\delta_{i+1}^{n}\right\}  $ with the error $f_{i}\left(  n\right)  $ .
Then, if $k$ is a constant of $n$, there is an asymptotic tangent simulation
of $\left\{  p_{k}^{n},\delta_{k}^{n}\right\}  $ by $\left\{  p_{1}^{n}%
,\delta_{1}^{n}\right\}  $ with the error $\sum_{i=1}^{k-1}f_{i}\left(
n\right)  $.
\end{proposition}

\begin{proof}
Obvious thus omitted.
\end{proof}

\subsection{Simulation of probabiltiy distribution family: a background from
decision theory}

Concept of simulation has been discussed in the field of statistical decision
theory in relation with the notion of sufficiency\thinspace\cite{Torgersen}.
Consider families $\mathcal{E=}\left\{  p_{\theta}\right\}  _{\theta\in\Theta
}$ , $\mathcal{F=}\left\{  q_{\theta}\right\}  _{\theta\in\Theta}$ of
probability distributions, and a function $e:\theta\rightarrow e\left(
\theta\right)  >0$. Also, let $\left(  D,\mathcal{D}\right)  $ be a decision
space. Then $\mathcal{F}$ is said to be $e$\textit{-deficient} relative to
$\mathcal{E}$ \ if, for any loss function $W_{\theta}$ with $\left\vert
W_{\theta}\left(  d\right)  \right\vert \leq1$ and for any decision function
$d:x\rightarrow d\left(  x\right)  \in D$, there is $d^{\prime}:y\rightarrow
d^{\prime}\left(  y\right)  \in D$ with
\begin{equation}
\int q_{\theta}\left(  y\right)  W_{\theta}\left(  d^{\prime}\left(  y\right)
\right)  \mathrm{d}\mu^{\prime}\leq\int p_{\theta}\left(  x\right)  W_{\theta
}\left(  d\left(  x\right)  \right)  \mathrm{d}\mu+e_{\theta}.
\label{e-deficient}%
\end{equation}
0-defficiency is simply called deficiency. The celebrated \textit{randomizing
criteria}, a necessary and sufficient condition for $e$-defficiency\textit{
}is the existence of\textit{ }$\Lambda$ with\textit{ }%
\[
\left\Vert p_{\theta}-\Lambda\left(  q_{\theta}\right)  \right\Vert \leq
e_{\theta}\text{.}%
\]
\ Especially, 0-deficiency is equivalent to that $Y\sim q_{\theta}$ is a
sufficient statistic of $\mathcal{E=}\left\{  p_{\theta}\right\}  _{\theta
\in\Theta}$. Thus, $e$-defficiency is an approximate version of suffuciency.

This randomizing criteria motivates our emphasis on simulation. Its `local'
version
\begin{align*}
\sup_{\theta}\left\Vert p_{\theta}-\Lambda\left(  q_{\theta}\right)
\right\Vert  &  =0\\
\left\Vert \frac{\partial p_{\theta}}{\partial\theta^{i}}-\Lambda\left(
\frac{\partial q_{\theta}}{\partial\theta^{i}}\right)  \right\Vert  &  \leq
e_{i}%
\end{align*}
is called \textit{local }$e$\textit{-deficiency at} $\theta$.

\subsection{Gaussian shift family}

\begin{proposition}
\label{prop:g-shift}Suppose $\left\{  p,\delta\right\}  =\left\{
\mathrm{N}\left(  \theta,\sigma^{2}\right)  ,\delta\mathrm{N}\left(
\theta,\sigma^{2}\right)  \right\}  $. Suppose also (M0), and (N0) holds. Then
we have \quad%
\[
g_{p}\left(  \delta\right)  =\frac{1}{\sigma^{2}}=J_{p}\left(  \delta\right)
.
\]

\end{proposition}

\begin{proof}
By an affine coordinate change of the data space $\Omega=%
\mathbb{R}
$, $\left\{  p,\delta\right\}  =\left\{  \mathrm{N}\left(  \theta
,\sigma\right)  ,\delta\mathrm{N}\left(  \theta,\sigma\right)  \right\}  $ is
transformed to $\left\{  q,\delta^{\prime}\right\}  =\left\{  \mathrm{N}%
\left(  0,1\right)  ,\frac{1}{\sigma^{2}}\delta\mathrm{N}\left(  0,1\right)
\right\}  $. Its inverse coordinate transform coordinate change of the data
space $\Omega$ sends $\left\{  q,\delta^{\prime}\right\}  $ to $\left\{
p,\delta\right\}  $. Therefore, by (M0) and (N0),
\[
g_{p}\left(  \delta\right)  =g_{\mathrm{N}\left(  0,1\right)  }\left(
\frac{1}{\sigma^{2}}\delta\mathrm{N}\left(  0,1\right)  \right)  =\frac
{1}{\sigma^{2}}g_{\mathrm{N}\left(  0,1\right)  }\left(  \delta\mathrm{N}%
\left(  0,1\right)  \right)  =\frac{1}{\sigma^{2}}.
\]
.
\end{proof}

\begin{remark}
Similarly, one can prove $\left\{  \mathrm{N}\left(  0,1\right)
,\delta\mathrm{N}\left(  0,1\right)  \right\}  ^{\otimes n}$, $\left\{
\mathrm{N}\left(  0,\frac{1}{n}\right)  ,\delta\mathrm{N}\left(  0,\frac{1}%
{n}\right)  \right\}  $, and $\left\{  \mathrm{N}\left(  0,1\right)  ,\sqrt
{n}\delta\mathrm{N}\left(  0,1\right)  \right\}  $ are equivalent.
\end{remark}

\subsection{On local asymptotic normality}

Asymptotic tangent simulation by Gaussian shift is somewhat analogous to
so-called \textit{local asymptotic normality} (LAN, in short) \cite{Strasser}.
Difference between them are as follows. First, asymptotic tangent simulation
is concerned only with a particular point $p$, while LAN is concerned also
with its neibourhood. On the other hand, (\ref{simulation-tangent-prob-1}) for
asymptotic tangent simulation is norm convergence, and thus obviously stronger
than convergence of $\frac{1}{\sqrt{n}}L^{\left(  n\right)  }$ to
$\mathrm{N}\left(  0,J\right)  $ in law.

\subsection{Zero bias transform}

Let $X$ be a real valued random variable with the distribution $\mathbb{P}%
_{X}$. Then,
\[
W_{X}\left(  x\right)  :=\frac{1}{\mathbb{V}\left(  X\right)  }\int_{-\infty
}^{x}\left(  \mathbb{E}X-y\right)  \mathbb{P}_{X}\left(  \mathrm{d}y\right)
\]
satisfies $\int W_{X}\left(  x\right)  \mathrm{d}y=1$, and thus defines a
random variable $X^{\circ}$. The map from $X$ to $X^{\circ}$ is called called
\textit{zero-bias transform \cite{CovTr}\cite{CovTr-2}\cite{CovTr-3}%
\cite{Zero-bias}\cite{Zero-bias-2}.} The following lemmas are proved in the
literatures mentioned above.

\begin{lemma}
Suppose $0<\mathbb{V}\left(  X\right)  <\infty$ and $\mathbb{E}\left(
X\right)  =0$. Suppose also $f:%
\mathbb{R}
\rightarrow%
\mathbb{R}
$ is absolutely continuous, differentiable, and $\mathbb{E}\left\vert
f^{\prime}\left(  X^{\circ}\right)  \right\vert <\infty\,$,
\begin{equation}
\mathbb{E}\left(  X\,f\left(  X\right)  \right)  =\mathbb{V}\left(  X\right)
\mathbb{E}f^{\prime}\left(  X^{\circ}\right)  . \label{cov-identity}%
\end{equation}

\end{lemma}

\begin{lemma}
\label{lem:sum-zero-bias}Let $S:=\sum_{i=1}^{n}X_{i}$, where $X_{1}%
,\cdots,X_{n}$ are IID with $0<\mathbb{V}\left(  X_{i}\right)  <\infty$ and
$\mathbb{E}\left(  X_{i}\right)  =0$. Then, denoting convolution by $\ast$,%
\begin{align*}
S^{\circ}  &  =S-X_{n}+X_{n}^{\circ},\\
W_{S}  &  =W_{X}\ast\left(  \mathrm{p}_{X}\right)  ^{\ast n-1}.
\end{align*}

\end{lemma}

\begin{lemma}
\label{lem:VW<VW}Let $S:=a_{1}X_{1}+a_{2}X_{2}$, where $a_{1}^{2}+a_{2}^{2}%
=1$. If $W_{X_{i}}\left(  x\right)  /\mathrm{p}_{X_{i}}\left(  x\right)
<\infty$ and $\mathbb{E}\left(  W_{X_{i}}\left(  X\right)  /\mathrm{p}_{X_{i}%
}\left(  X\right)  -1\right)  ^{2}<\infty$ ($i=1,2$),
\[
\mathbb{E}\left(  W_{S}\left(  S\right)  /\mathrm{p}_{S}\left(  S\right)
-1\right)  ^{2}\leq a_{1}^{4}\mathbb{E}\left(  W_{X_{1}}\left(  X_{1}\right)
/\mathrm{p}_{X_{1}}\left(  X_{1}\right)  -1\right)  ^{2}+a_{2}^{4}%
\mathbb{E}\left(  W_{X_{2}}\left(  X_{2}\right)  /\mathrm{p}_{X_{2}}\left(
X_{2}\right)  -1\right)  ^{2}%
\]

\end{lemma}

\begin{lemma}
\label{lem:zero-bias-support}The random variable $X^{\circ}$ is supported on a
subset of the convex hull of the support of $X$.
\end{lemma}

\subsection{Binary distributions}

\label{subsec:binary}

Consider a family of binary distributions $\left\{  p_{\theta}\right\}  $,
where the data space is $\Omega=\left\{  0,1\right\}  $. Letting $N_{1}\left(
x^{n}\right)  $ be the number of $1$ in the sequence $x^{n}=x_{1}x_{2}\cdots
x_{n}$,%
\begin{align*}
L^{\left(  n\right)  }  & =N_{1}\left(  x^{n}\right)  \left\{  L\left(
1\right)  -L\left(  0\right)  \right\}  +nL\left(  0\right)  \\
& =\alpha\left\{  N_{1}\left(  x^{n}\right)  -np\left(  1\right)  \right\}  ,
\end{align*}
where $\alpha:=L\left(  1\right)  -L\left(  0\right)  $.

We compose $\tilde{\Lambda}^{n}$ which satisfies
(\ref{simulation-tangent-prob-1-2}) with $\left\{  \mathrm{p}_{L^{\left(
n\right)  }},\,L^{\left(  n\right)  }\mathrm{p}_{L^{\left(  n\right)  }%
}\right\}  \equiv\left\{  p^{\otimes n},\delta^{\left(  n\right)  }\right\}
:=\left\{  p_{\theta}^{\otimes n},\delta_{\theta}^{\left(  n\right)
}\right\}  $ and
\begin{align*}
\left\{  q^{n},\delta^{\prime n}\right\}   &  :=\left\{  \mathrm{N}\left(
0,nJ_{p}\left(  \delta\right)  \right)  ,nJ_{p}\left(  \delta\right)
\delta\mathrm{N}\left(  0,nJ_{p}\left(  \delta\right)  \right)  \right\}  \\
&  \equiv\left\{  \mathrm{N}\left(  0,1\right)  ,\sqrt{nJ_{p}\left(
\delta\right)  }\delta\mathrm{N}\left(  0,1\right)  \right\}  \equiv\left\{
\mathrm{N}\left(  0,1\right)  ,\delta\mathrm{N}\left(  0,1\right)  \right\}
^{\otimes nJ_{p}\left(  \delta\right)  },
\end{align*}
\ by letting $\tilde{L}^{n}$ be the element of$\ $the set
\[
\left\{  \alpha\left(  n_{1}-p\left(  1\right)  n\right)  \,;n_{1}\in%
\mathbb{N}
,n_{1}\leq n\right\}
\]
closest to $L^{\prime}{}^{n}{}\sim\mathrm{N}\left(  0,nJ_{p}\left(
\delta\right)  \right)  $.

One can easily verify
\begin{align*}
\mathbb{E}\left\vert \mathbb{E}\left[  L^{\prime n}\,|\tilde{L}^{n}\,\right]
-\tilde{L}^{n}\right\vert  &  \leq\left\vert \alpha_{\theta}\right\vert ,\\
\mathbb{E}\left(  L\right)  ^{2} &  =J_{p}\left(  \delta\right)  <\infty,\\
\frac{1}{n}\mathbb{E}\left(  \tilde{L}^{n}\right)  ^{2} &  \leq\frac{1}%
{n}\mathbb{E}\left(  L^{\prime n}\right)  ^{2}+\frac{1}{n}\,\mathbb{E}%
\left\vert L^{\prime n}\,-\tilde{L}^{n}\right\vert ^{2}\\
&  \leq J_{p}\left(  \delta\right)  +\frac{1}{n}\left(  \alpha\right)  ^{2}.
\end{align*}

(\ref{simulation-tangent-prob-1-2}), or
\[
\left\Vert p^{\otimes n}-\Lambda\left(  q^{n}\right)  \right\Vert _{1}%
\leq\left\Vert \mathrm{p}_{L^{\left(  n\right)  }}-\mathrm{p}_{\tilde{L}^{n}%
}\right\Vert _{1}\leq\frac{1}{\sqrt{n}}\frac{4}{\sqrt{J_{p}\left(
\delta\right)  }},
\]
is the direct consequence of Theorem\thinspace\ref{th:CLTforBin} below. Hence,
by Lemma\thinspace\ref{lem:sim-tangent}, the error of this tangent simulation
is $\frac{A}{n^{1/4}}$ with
\[
A=\frac{8}{\sqrt{J_{p}\left(  \delta\right)  }}+3J_{p}\left(  \delta\right)
+3\left(  \alpha\right)  ^{2}+3\left\vert \alpha\right\vert ,
\]
which is continuous function of $\delta\left(  0\right)  $ and $p\left(
0\right)  $ is bounded on any compact region.

\begin{theorem}
\label{th:CLTforBin}Let $X_{1},X_{2},\cdots,X_{n}$ be the IID random variables
taking values in $\left\{  0,1\right\}  $, with $\Pr\left\{  X_{1}=1\right\}
=\eta$. Denote its variance by $\sigma^{2}$, and define
\[
Y_{i}:=\frac{1}{\sqrt{n}\sigma}\left(  X_{i}-\eta\right)  ,\,\quad S_{n}%
:=\sum_{i=1}^{n}Y_{i}.
\]
Suppose
\[
\mathcal{A}=\bigcup_{z}\mathcal{A}_{z}^{n},
\]
where
\[
\mathcal{A}_{z}^{n}:=\left[  z-\frac{1}{2\sqrt{n}\sigma},z+\frac{1}{2\sqrt
{n}\sigma}\right]
\]
and $z$ runs over a subset of $\left(
\mathbb{Z}
-n\eta\right)  /\sqrt{n}\sigma$. Then,
\[
\left\vert \Pr\left\{  S_{n}\in\mathcal{A}\right\}  -\Pr\left\{
\mathrm{N}\left(  0,1\right)  \in\mathcal{A}\right\}  \right\vert \leq\frac
{1}{\sqrt{n}\sigma}.
\]

\end{theorem}

\begin{proof}
Letting
\begin{equation}
\psi_{\mathcal{A}}\left(  x\right)  :=e^{\frac{x^{2}}{2}}\int_{-\infty}%
^{x}\left(  \chi_{\mathcal{A}}\left(  t\right)  -\Pr\left\{  \mathrm{N}\left(
0,1\right)  \in\mathcal{A}\right\}  \right)  e^{-\frac{t^{2}}{2}}\mathrm{d}t,
\label{def-psi-1}%
\end{equation}
we have
\begin{align}
&  \left\vert \Pr\left\{  S_{n}\in\mathcal{A}\right\}  -\Pr\left\{
\mathrm{N}\left(  0,1\right)  \in\mathcal{A}\right\}  \right\vert
\underset{(i)}{=}\left\vert \mathbb{E}\left(  \frac{\mathrm{d}}{\mathrm{d}%
x}\psi_{\mathcal{A}}\left(  S_{n}\right)  -S_{n}\psi_{\mathcal{A}}\left(
S_{n}\right)  \right)  \right\vert \nonumber\\
&  \underset{(ii)}{=}\left\vert \mathbb{E}\left(  \frac{\mathrm{d}}%
{\mathrm{d}x}\psi_{\mathcal{A}}\left(  S_{n}\right)  -\frac{\mathrm{d}%
}{\mathrm{d}x}\psi_{\mathcal{A}}\left(  S_{n}^{\circ}\right)  \right)
\right\vert \text{ }\nonumber\\
&  \underset{(iii)}{=}\left\vert \mathbb{E}\left(  \chi_{\mathcal{A}}\left(
S_{n}\right)  -\chi_{\mathcal{A}}\left(  S_{n}^{\circ}\right)  \right)
+\mathbb{E}\left(  S_{n}\psi_{\mathcal{A}}\left(  S_{n}\right)  -S_{n}^{\circ
}\psi_{\mathcal{A}}\left(  S_{n}^{\circ}\right)  \right)  \right\vert
\nonumber\\
&  \underset{(iv)}{=}2\left\vert \mathbb{E}\left(  \chi_{\mathcal{A}}\left(
S_{n}\right)  -\chi_{\mathcal{A}}\left(  S_{n}^{\circ}\right)  \right)
\right\vert ,\nonumber\\
&  \underset{(v)}{=}2\left\vert \mathbb{E}\left(  \chi_{\mathcal{A}}\left(
S_{n}\right)  -\chi_{\mathcal{A}}\left(  S_{n}-Y_{n}+Y_{n}^{\circ}\right)
\right)  \right\vert . \label{Stein-1}%
\end{align}
where $(i)$ and $(iii)$ are due to the definition (\ref{def-psi-1}), $(ii)$ is
due to (\ref{cov-identity}), $(iv)$ is due to
\[
\psi_{\mathcal{A}}\left(  x\right)  \leq\frac{1}{\left\vert x\right\vert },
\]
and $(v)$ is due to Lemma\thinspace\ref{lem:sum-zero-bias}. By definition, one
can verify that $X_{i}^{\circ}\sim W_{X_{i}}$ is uniform distirbution over
$[0,1]$. Also,
\begin{align*}
&  \left\vert \mathbb{E}\left(  \chi_{\mathcal{A}}\left(  S_{n}\right)
-\chi_{\mathcal{A}}\left(  S_{n}-Y_{n}+Y_{n}^{\circ}\right)  \right)
\right\vert \\
&  \leq\sum_{k=0}^{n}\left\vert \Pr\left\{  \sum_{i=1}^{n-1}X_{i}%
+X_{n}=k\right\}  -\Pr\left\{  \sum_{i=1}^{n-1}X_{i}+X_{n}^{\circ}\in\left[
k-\frac{1}{2},k+\frac{1}{2}\right]  \right\}  \right\vert \\
&  =\sum_{k=0}^{n}\left\vert \Pr\left\{  \sum_{i=1}^{n-1}X_{i}=k\right\}
\left\{  \left(  1-\eta\right)  -\frac{1}{2}\right\}  +\Pr\left\{  \sum
_{i=1}^{n-1}X_{i}=k-1\right\}  \left(  \eta-\frac{1}{2}\right)  \right\vert \\
&  =\sum_{k=0}^{n}\Pr\left\{  \sum_{i=1}^{n}X_{i}=k\right\}  \left\vert
\frac{1}{1-\eta}\frac{n-k}{n}\left\{  \left(  1-\eta\right)  -\frac{1}%
{2}\right\}  +\frac{1}{\eta}\frac{k}{n}\left(  \eta-\frac{1}{2}\right)
\right\vert \\
&  =\frac{\left\vert \eta-\frac{1}{2}\right\vert }{\eta\left(  1-\eta\right)
}\sum_{k=0}^{n}\Pr\left\{  \sum_{i=1}^{n}X_{i}=k\right\}  \left\vert \frac
{k}{n}-\eta\right\vert \\
&  \leq\frac{\left\vert \eta-\frac{1}{2}\right\vert }{\eta\left(
1-\eta\right)  }\sqrt{\sum_{k=0}^{n}\Pr\left\{  \sum_{i=1}^{n}X_{i}=k\right\}
\left(  \frac{k}{n}-\eta\right)  ^{2}}=\frac{\left\vert \eta-\frac{1}%
{2}\right\vert }{\left\{  \eta\left(  1-\eta\right)  \right\}  ^{1/2}}\frac
{1}{\sqrt{n}},
\end{align*}
which leads to the assertion.
\end{proof}

\subsection{Distributions over the finite set}

\begin{theorem}
\label{th:finite-sim-prob}Suppose $p$ is a probability distribution and
$\delta$ is a signed measure over a set $\Omega$ with $\left\vert
\Omega\right\vert =k$ ($k<\infty$). Let $J:=J_{p}\left(  \delta\right)  $,
$\varepsilon>0$ and
\[
\left\{  q^{n},\delta^{\prime n}\right\}  :=\left\{  \mathrm{N}\left(
0,1\right)  ,\sqrt{n\left(  J+\varepsilon\right)  }\delta\mathrm{N}\left(
0,1\right)  \right\}  \equiv\left\{  \mathrm{N}\left(  0,1\right)
,\delta\mathrm{N}\left(  0,1\right)  \right\}  ^{\otimes n\left(
J+\varepsilon\right)  }.
\]
Then, we can compose $\Lambda^{n}$ with the error $\frac{A}{n^{1/4}}$, where
$A$ is a continuous function of $\left\{  p\left(  x\right)  ,\delta\left(
x\right)  ;x=1,\cdots,k-1\right\}  $ and is bounded on any compact region. 
\end{theorem}

\begin{proof}
Since binary distributions can be simulated by Gaussian shift as in
Subsection\thinspace\ref{subsec:binary}, due to Propositin\thinspace
\ref{prop:sym-chain}, we only have to compose asymptotic tangent simulation
$\left\{  p,\delta\right\}  ^{\otimes n}$ by binary distributions. For that,
we first asymptotically simulate $\left\{  p,\delta\right\}  ^{\otimes n}$ by
$\left\{  p_{a},\delta_{a}\right\}  ^{\otimes n}\otimes\left\{  p_{A}%
,\delta_{A}\right\}  ^{\otimes n_{a}}$, where $\left\{  p_{a},\delta
_{a}\right\}  $ and $\left\{  p_{A},\delta_{A}\right\}  $ is defined over the
binary set $\Omega_{a}$ and and the set $\Omega_{A}$ with $\left(  k-1\right)
$-elements, respectively.Then, by virtue of Propositin\thinspace
\ref{prop:sym-chain}, inductive argument leads to asymptotic tangent
simulation by binary distributions.

Let
\begin{align*}
p_{a}\left(  0\right)   &  :=p\left(  k\right)  ,\,\,p_{a}\left(  1\right)
=\sum_{x=1}^{k-1}p\left(  x\right)  ,\\
\delta_{a}\left(  0\right)   &  :=\delta\left(  k\right)  ,\,\,\delta
_{a}\left(  1\right)  =\sum_{x=1}^{k-1}\delta\left(  x\right)  ,\,L_{a}%
:=\frac{\delta_{a}}{p_{a}}\\
p_{A}\left(  x\right)   &  :=\frac{p\left(  x\right)  }{p_{a}\left(  1\right)
},\,(x=1,\cdots,k-1),\\
\,\delta_{A}\left(  x\right)   &  :=\frac{\delta\left(  x\right)  }%
{p_{a}\left(  1\right)  }-\frac{\delta_{a}\left(  1\right)  p_{A}\left(
x\right)  }{p_{a}\left(  1\right)  },\,(x=1,\cdots,k-1),\\
L_{A}  &  :=\frac{\delta_{A}\left(  x\right)  }{p_{A}\left(  x\right)
},\,\ \,(x=1,\cdots,k-1),\\
n_{a}  &  :=n\left(  p_{a}\left(  1\right)  +\varepsilon\right)  \text{
\thinspace\thinspace}(\varepsilon>0).
\end{align*}
Also, let $x_{a}^{n}=x_{a1}x_{a2}\cdots x_{an}\in\Omega_{a}^{\otimes n}$,
$x_{A}^{n}=x_{A1}x_{A2}\cdots x_{An}\in\Omega_{A}^{\otimes n}$, $X_{ai}\sim
p_{a}$, $X_{Ai}\sim p_{A}$, $X_{a}^{n}\sim p_{a}^{\otimes n}$, and $X_{A}%
^{n}\sim p_{A}^{\otimes n}$. Denote by $N_{1}\left(  x_{a}^{n}\right)  $ the
number of $1$ in the sequence $x_{a}^{n}=x_{a1}x_{a2}\cdots x_{an}$. \ Also,
we identify the pair $\left(  x_{a},x_{A}\right)  $ with $x$, by the
correspondence
\[
x\equiv\left\{
\begin{array}
[c]{cc}%
\left(  1,x\right)  & \left(  x=1,\cdots,k-1\right)  ,\\
\left(  0,\#\right)  & x=k,
\end{array}
\right.  \,\,
\]
where $\#$ stands for empty string. To define asymptotic tangent simulation,
one define function $F:\Omega^{\otimes n}\rightarrow\Omega^{\otimes n}$ such
that
\[
F\left(  x^{n}\right)  =\left\{
\begin{array}
[c]{cc}%
x^{n} & \left(  N_{1}\left(  x_{a}^{n}\right)  \leq n_{a}\right)  ,\\
k^{n} & \left(  N_{1}\left(  x_{a}^{n}\right)  >n_{a}\right)  .
\end{array}
\right.
\]
Using $F$, we define
\begin{align*}
\Lambda^{n}\left(  r^{n}\right)  \left(  x^{n}\right)   &  :=\sum_{y^{n}\in
F^{-1}\left(  x^{n}\right)  }r^{n}\left(  y^{n}\right)  ,\\
\tilde{p}^{n}  &  :=\Lambda^{n}\left(  p^{\otimes n}\right)  ,\,\,\,\tilde
{\delta}^{n}:=\Lambda^{n}\left(  \delta^{\left(  n\right)  }\right)  .
\end{align*}
Then,
\begin{align*}
\left\Vert \tilde{p}^{n}-p^{\otimes n}\right\Vert _{1}  &  =\sum_{x^{n}%
}\left\vert \sum_{y^{n}\in F^{-1}\left(  x^{n}\right)  }p^{\otimes n}\left(
y^{n}\right)  -p^{\otimes n}\left(  x^{n}\right)  \right\vert \\
&  =2\sum_{x^{n}:N_{1}\left(  x_{a}^{n}\right)  >n_{a}}p^{\otimes n}\left(
x^{n}\right) \\
&  \leq2\exp\left\{  -nC_{a,\varepsilon}\right\}  ,
\end{align*}
where%
\[
C_{a,\varepsilon}:=p_{a}\left(  0\right)  \ln\frac{p_{a}\left(  0\right)
}{p_{a}\left(  0\right)  -\varepsilon}+p_{a}\left(  1\right)  \ln\frac
{p_{a}\left(  1\right)  }{p_{a}\left(  1\right)  +\varepsilon},
\]
and
\begin{align*}
\frac{1}{\sqrt{n}}\left\Vert \tilde{\delta}^{n}-\delta^{\left(  n\right)
}\right\Vert _{1}  &  =\frac{1}{\sqrt{n}}\left\vert \sum_{N_{1}\left(
x_{a}^{n}\right)  >n_{a}}\delta^{\left(  n\right)  }\left(  x^{n}\right)
\right\vert +\frac{1}{\sqrt{n}}\sum_{N_{1}\left(  x_{a}^{n}\right)  >n_{a}%
}\left\vert \delta^{\left(  n\right)  }\left(  x^{n}\right)  \right\vert \\
&  \leq\frac{2}{\sqrt{n}}\sum_{N_{1}\left(  x_{a}^{n}\right)  >n_{a}%
}\left\vert \delta^{\left(  n\right)  }\left(  x^{n}\right)  \right\vert \\
&  =\frac{2}{\sqrt{n}}\sum_{x_{a}^{n}:N_{1}\left(  x_{a}^{n}\right)  >n_{a}%
}\sum_{x_{A}^{N_{A}\left(  x_{a}^{n}\right)  }}\left\vert L_{a}^{\left(
n\right)  }\left(  x_{a}^{n}\right)  +L_{A}^{\left(  N_{1}\left(  x_{a}%
^{n}\right)  \right)  }\left(  x_{A}^{N_{1}\left(  x_{a}^{n}\right)  }\right)
\right\vert p_{a}^{\otimes n}\left(  x_{a}^{n}\right)  p_{A}^{\otimes
N_{A}\left(  x_{a}^{n}\right)  }\left(  x_{A}^{N_{A}\left(  x_{a}^{n}\right)
}\right) \\
&  \leq\frac{2}{\sqrt{n}}\sum_{x_{a}^{n}:N_{1}\left(  x_{a}^{n}\right)
>n_{a}}\left[  n\max\left\{  \left\vert L_{a}\left(  0\right)  \right\vert
,\left\vert L_{a}\left(  1\right)  \right\vert \right\}  +n\left\vert
\max_{1\leq x_{A1}\leq k-1}L_{A}\left(  x_{A1}\right)  \right\vert \right]
p_{a}^{\otimes n}\left(  x_{a}^{n}\right) \\
&  \leq2\sqrt{n}\left[  \max\left\{  \left\vert L_{a}\left(  0\right)
\right\vert ,\left\vert L_{a}\left(  1\right)  \right\vert \right\}
+\left\vert \max_{1\leq x\leq k-1}L_{A}\left(  x\right)  \right\vert \right]
\exp\left\{  -nC_{a,\varepsilon}\right\}  .
\end{align*}
Also,%
\begin{align*}
J_{p}\left(  \delta\right)   &  =\frac{\left\{  \delta\left(  k\right)
\right\}  ^{2}}{p\left(  k\right)  }+\sum_{x=1}^{k-1}\frac{\left\{
\delta\left(  x\right)  \right\}  ^{2}}{p\left(  x\right)  }\\
&  =\frac{\left\{  \delta_{a}\left(  0\right)  \right\}  ^{2}}{p_{a}\left(
0\right)  }+\sum_{x=1}^{k-1}\frac{1}{p_{a}\left(  1\right)  p_{A}\left(
x\right)  }\left\{  p_{a}\left(  1\right)  \delta_{A}\left(  x\right)
+\delta_{a}\left(  1\right)  p_{A}\left(  x\right)  \right\}  ^{2}\\
&  =\frac{\left\{  \delta_{a}\left(  0\right)  \right\}  ^{2}}{p_{a}\left(
0\right)  }+\frac{\left\{  \delta_{a}\left(  1\right)  \right\}  ^{2}}%
{p_{a}\left(  1\right)  }+p_{a}\left(  1\right)  \sum_{x=1}^{k-1}%
\frac{\left\{  \delta_{A}\left(  x\right)  \right\}  ^{2}}{p_{A}\left(
x\right)  }+\delta_{a}\left(  1\right)  \sum_{x=1}^{k-1}\delta_{A}\left(
x\right) \\
&  =J_{a}\left(  \delta_{a}\right)  +p_{a}\left(  1\right)  J_{A}\left(
\delta_{A}\right)  .
\end{align*}

\end{proof}

Analogously, one can compose an asymptotic tangent simulation of $\left\{
p_{A},\delta_{A}\right\}  ^{\otimes n_{a}}$ by $\left\{  p_{b},\delta
_{b}\right\}  ^{\otimes n_{a}}\otimes\left\{  p_{B},\delta_{B}\right\}
^{\otimes n_{b}}$, where $\left\{  p_{b},\delta_{b}\right\}  $ and $\left\{
p_{B},\delta_{B}\right\}  $ are defined over the binary set $\Omega_{b}$ and
and the set $\Omega_{B}$ with $\left(  k-2\right)  $-elements , respectively,
where
\begin{align*}
p_{b}\left(  0\right)   &  :=p_{A}\left(  k-1\right) \\
p_{b}\left(  1\right)   &  :=\sum_{x=1}^{k-2}p_{A}\left(  x\right) \\
n_{b}  &  :=n_{a}\left(  p_{b}\left(  1\right)  +\varepsilon\right)
,\text{\thinspace}\\
J_{p_{A}}\left(  \delta_{A}\right)   &  =J_{p_{b}}\left(  \delta_{b}\right)
+p_{b}\left(  1\right)  J_{p_{B}}\left(  \delta_{B}\right)  .
\end{align*}
Repeating this proscess recursively, by Proposition\thinspace
\ref{prop:sym-chain}, one can asymptotically simulate $\left\{  p,\delta
\right\}  ^{\otimes n}$ by
\begin{equation}
\left\{  p_{a},\delta_{a}\right\}  ^{\otimes n}\otimes\left\{  p_{b}%
,\delta_{b}\right\}  ^{\otimes n_{a}}\otimes\cdots\otimes\left\{  p_{z}%
,\delta_{z}\right\}  ^{\otimes n_{y}}, \label{prod:binary}%
\end{equation}
($\left\{  p_{Z},\delta_{Z}\right\}  $ is defined over $\left\{  1\right\}  $,
thus is trivia)l with the error
\[
2\sqrt{n}\left[  \sum_{i=a}^{z}\max\left\{  \left\vert L_{i}\left(  0\right)
\right\vert ,\left\vert L_{i}\left(  1\right)  \right\vert \right\}
+\sum_{j=A}^{Y}\left\vert \max_{1\leq x\leq k-1}L_{j}\left(  x\right)
\right\vert +1\right]  \sum_{i=a}^{z}\exp\left\{  -nC_{i,\varepsilon}\right\}
,
\]
which is upperbouded by $\frac{B}{n^{1/4}}$, where $B$ is a continuous
function of \ $\left\{  p\left(  x\right)  ,\delta\left(  x\right)
;x=1,\cdots,k-1\right\}  $. \ Due to Subsection\thinspace\ref{subsec:binary}%
\ , (\ref{prod:binary}) can be simulated by
\begin{align*}
&  \left\{  \mathrm{N}\left(  0,1\right)  ,\delta\mathrm{N}\left(  0,1\right)
\right\}  ^{\otimes nJ_{a}\left(  \delta_{a}\right)  }\otimes\left\{
\mathrm{N}\left(  0,1\right)  ,\delta\mathrm{N}\left(  0,1\right)  \right\}
^{\otimes n_{a}J_{b}\left(  \delta_{b}\right)  }\otimes\\
&  \cdots\otimes\left\{  \mathrm{N}\left(  0,1\right)  ,\delta\mathrm{N}%
\left(  0,1\right)  \right\}  ^{\otimes n_{y}J_{z}\left(  \delta_{z}\right)
}\\
&  \equiv\left\{  \mathrm{N}\left(  0,1\right)  ,\delta\mathrm{N}\left(
0,1\right)  \right\}  ^{\otimes n\left(  J_{p}\left(  \delta\right)  +f\left(
\varepsilon\right)  \right)  },
\end{align*}
where $\lim_{\varepsilon\rightarrow o}f\left(  \varepsilon\right)  =0$, with
the error $\frac{B^{\prime}}{n^{1/4}}$, where $B^{\prime}$ is a continuous
function of \ $\left\{  p\left(  x\right)  ,\delta\left(  x\right)
;x=1,\cdots,k-1\right\}  $. Here, `$\equiv$' is due to
\begin{align*}
&  nJ_{a}\left(  \delta_{a}\right)  +n_{a}J_{b}\left(  \delta_{b}\right)
+n_{b}J_{c}\left(  \delta_{c}\right)  +\cdots+n_{y}J_{z}\left(  \delta
_{z}\right) \\
&  =nJ_{a}\left(  \delta_{a}\right)  +n\left(  p_{a}\left(  1\right)
+\varepsilon\right)  J_{b}\left(  \delta_{b}\right)  +n\left(  p_{a}\left(
1\right)  +\varepsilon\right)  \left(  p_{b}\left(  1\right)  +\varepsilon
\right)  J_{c}\left(  \delta_{c}\right) \\
&  +\cdots+n\prod_{i=a}^{y}\left(  p_{i}\left(  1\right)  +\varepsilon\right)
J_{z}\left(  \delta_{z}\right) \\
&  =n\left(  J_{p}\left(  \delta\right)  +f\left(  \varepsilon\right)
\right)
\end{align*}
where the last identiy is due to
\[
J_{p}\left(  \delta\right)  =J_{a}\left(  \delta_{a}\right)  +p_{a}\left(
1\right)  J_{b}\left(  \delta_{b}\right)  +p_{a}\left(  1\right)  p_{b}\left(
1\right)  J_{c}\left(  \delta_{c}\right)  +\cdots+\prod_{i=a}^{y}p_{i}\left(
1\right)  J_{z}\left(  \delta_{z}\right)  .
\]
Therefore, due to Proposition\thinspace\ref{prop:sym-chain}, we obtain an
asymptotic tangent simulation of $\left\{  p,\delta\right\}  ^{\otimes n}$ by
$\left\{  q^{n},\delta^{\prime n}\right\}  $ with the error $\frac
{B+B^{\prime}}{n^{1/4}}$, and the assertion is proved.

\subsection{ A continuous random variable with smooth density}

\begin{theorem}
\label{th:cont-prob-sim}Let $\Omega=%
\mathbb{R}
$. Suppose $L\left(  x\right)  =\delta\left(  x\right)  /p\left(  x\right)  $
exists and is a continuous function of $x$. Let $J:=J_{p}\left(
\delta\right)  $, $\left\{  q^{n},\delta^{\prime n}\right\}  :=\left\{
\mathrm{N}\left(  0,1\right)  ,\sqrt{nJ}\delta\mathrm{N}\left(  0,1\right)
\right\}  =\left\{  \mathrm{N}\left(  0,1\right)  ,\delta\mathrm{N}\left(
0,1\right)  \right\}  ^{\otimes nJ}$ , and define $\tilde{L}^{n}%
=\Lambda^{n\ast}\left(  L^{\prime n}\right)  :=L^{\prime n}$. Suppose
\begin{equation}
\mathbb{E}\left(  \frac{\int_{-\infty}^{L}\left(  -t\right)  \,\mathrm{p}%
_{L}\left(  t\right)  \mathrm{d}t}{J\,\mathrm{p}_{L}\left(  L\right)
}\right)  ^{2}<\infty, \label{w-variance}%
\end{equation}
holds. Then, $\left\{  q^{n},\delta^{\prime n},\Lambda^{n}\right\}  $
\ satisfies (\ref{simulation-tangent-prob-1-2}) and
(\ref{simulation-tangent-prob-2-2}). Thus, by Lemma\thinspace
\ref{lem:sim-tangent}, it satisfies (\ref{simulation-tangent-prob-1}) and
(\ref{simulation-tangent-prob-2}).
\end{theorem}

\begin{proof}
The assertion is essentially the same as Theorem\thinspace2.3 of
\cite{CovTr-2}. For the sake of completeness, however, the whole argument is
described below. Let $S_{n}:=\frac{1}{\sqrt{nJ}}L^{\left(  n\right)  }$. In
the same way as the proof of Theorem\thinspace\ref{th:CLTforBin}, we have
\begin{align*}
&  \left\vert \Pr\left\{  S_{n}\in\mathcal{A}\right\}  -\Pr\left\{
\mathrm{N}\left(  0,1\right)  \in\mathcal{A}\right\}  \right\vert \\
&  =\left\vert \mathbb{E}\left(  \chi_{\mathcal{A}}\left(  S_{n}\right)
-\chi_{\mathcal{A}}\left(  S_{n}^{\circ}\right)  \right)  +\mathbb{E}\left(
S_{n}\psi_{\mathcal{A}}\left(  S_{n}\right)  -S_{n}^{\circ}\psi_{\mathcal{A}%
}\left(  S_{n}^{\circ}\right)  \right)  \right\vert \\
&  \leq2\left\Vert \mathrm{p}_{S_{n}}-\mathrm{p}_{S_{n}^{\circ}}\right\Vert
_{1}=2\mathbb{E}\left\vert 1-\frac{W_{S_{n}}\left(  S_{n}\right)  }%
{\mathrm{p}_{S_{n}}\left(  S_{n}\right)  }\right\vert \leq2\sqrt
{\mathbb{E}\left(  \frac{W_{S_{n}}\left(  S_{n}\right)  }{\mathrm{p}_{S_{n}%
}\left(  S_{n}\right)  }-1\right)  ^{2}}.
\end{align*}
where $\psi_{\mathcal{A}}$ is defined by (\ref{def-psi-1}), and thus
$\left\vert x\psi_{\mathcal{A}}\left(  x\right)  \right\vert \leq1$. Hence, it
boils down to the evaluation of $\mathbb{E}\left(  \frac{W_{S_{n}}\left(
S_{n}\right)  }{\mathrm{p}_{S_{n}}\left(  S_{n}\right)  }-1\right)  ^{2}$,
which, due to Lemma\thinspace\ref{lem:VW<VW}, is not larger than
\begin{align*}
&  \frac{1}{\sqrt{n}}\mathbb{E}\left(  \frac{W_{\frac{1}{\sqrt{J}}L}\left(
\frac{1}{\sqrt{J}}L\right)  }{\mathrm{p}_{\frac{1}{\sqrt{J}}L}\left(  \frac
{1}{\sqrt{J}}L\right)  }-1\right)  ^{2}=\frac{1}{\sqrt{n}}\mathbb{E}\left(
\frac{\int_{-\infty}^{\frac{1}{\sqrt{J}}L}\left(  -t\right)  \,\mathrm{p}%
_{\frac{1}{\sqrt{J}}L}\left(  t\right)  \mathrm{d}t}{\mathrm{p}_{\frac
{1}{\sqrt{J}}L}\left(  \frac{1}{\sqrt{J}}L\right)  }-1\right)  ^{2}\\
&  =\frac{1}{\sqrt{n}}\mathbb{E}\left(  \frac{\int_{-\infty}^{L}\left(
-\frac{1}{\sqrt{J}}t\right)  \,\mathrm{p}_{L}\left(  t\right)  \mathrm{d}%
t}{\sqrt{J}\mathrm{p}_{L}\left(  l\right)  }-1\right)  ^{2}=\frac{1}{\sqrt{n}%
}\mathbb{E}\left(  \frac{\int_{-\infty}^{L}\left(  -t\right)  \,\mathrm{p}%
_{L}\left(  t\right)  \mathrm{d}t}{J\,\mathrm{p}_{L}\left(  l\right)
}-1\right)  ^{2}\\
&  =\frac{1}{\sqrt{n}}\left\{  \mathbb{E}\left(  \frac{\int_{-\infty}%
^{L}\left(  -t\right)  \,\mathrm{p}_{L}\left(  t\right)  \mathrm{d}%
t}{J\,\mathrm{p}_{L}\left(  l\right)  }\right)  ^{2}-1\right\}  .
\end{align*}
Hence, we have (\ref{simulation-tangent-prob-1-2}). Also, it is easy to
verify
\begin{align*}
\mathbb{E}\left\vert \mathbb{E}\left[  L^{\prime n}\,|\tilde{L}^{n}\,\right]
-\tilde{L}^{n}\right\vert  &  =0,\\
\mathbb{E}\left(  L\right)  ^{2}  &  =\frac{1}{n}\mathbb{E}\left(  \tilde
{L}^{n}\right)  ^{2}=J_{p}\left(  \delta\right)  .
\end{align*}

\end{proof}

A trivial sufficient condition for (\ref{w-variance}) is that the support of
$\mathrm{p}_{L}$ is bounded. Also, suppose
\begin{align*}
\frac{a_{1}}{t^{\alpha_{1}}}  &  \leq\mathrm{p}_{L}\left(  t\right)  \leq
\frac{b_{1}}{t^{\alpha_{1}}},\,\left(  t\leq\exists t_{1}\right) \\
\frac{a_{2}}{t^{\alpha_{2}}}  &  \leq\mathrm{p}_{L}\left(  t\right)  \leq
\frac{b_{2}}{t^{\alpha_{2}}},\left(  t\geq\exists t_{2}\right)
\end{align*}
hold for some real constant $a_{i}$, $b_{i}$, $\alpha_{i}$ ($\iota=1$,$2$).
Then, if $y<t_{1}$,
\[
\frac{1}{\mathrm{p}_{L}\left(  y\right)  }\int_{-\infty}^{y}\left(  -t\right)
\mathrm{p}_{L}\left(  t\right)  \mathrm{d}t\leq\frac{1}{\alpha_{1}-2}%
\frac{b_{1}}{a_{1}}y^{2},
\]
and, due to $\int_{-\infty}^{\infty}\left(  -t\right)  \mathrm{p}_{L}\left(
t\right)  \mathrm{d}t=0$, if $y>t_{2}$,
\[
\frac{1}{\mathrm{p}_{L}\left(  y\right)  }\int_{-\infty}^{y}\left(  -t\right)
\mathrm{p}_{L}\left(  t\right)  \mathrm{d}t=\frac{1}{\mathrm{p}_{L}\left(
y\right)  }\int_{y}^{\infty}t\,\mathrm{p}_{L}\left(  t\right)  \mathrm{d}%
t\leq\frac{1}{\alpha_{2}-2}\frac{b_{2}}{a_{2}}y^{2}.
\]
Hence, if
\[
\min\left\{  \alpha_{1},\alpha_{2}\right\}  \geq4\text{,}%
\]
we have (\ref{w-variance}).

The following conditions are also sufficient:%
\begin{align}
a_{1}e^{-\left\vert t\right\vert ^{\alpha_{1}}}  &  \leq\mathrm{p}_{L}\left(
t\right)  \leq b_{1}e^{-\left\vert t\right\vert ^{\alpha_{1}}},\,\left(
t\leq\exists t_{1}\right)  ,\nonumber\\
a_{2}e^{-\left\vert t\right\vert ^{\alpha_{2}}}  &  \leq\mathrm{p}_{L}\left(
t\right)  \leq b_{2}e^{-\left\vert t\right\vert ^{\alpha_{2}}},\left(
t\geq\exists t_{2}\right)  \label{exp-tail}%
\end{align}
for some real constants $a_{i}$, $b_{i}$, and $\alpha_{i}$ ($i=1$,$2$) , with
\[
\min\left\{  \alpha_{1},\alpha_{2}\right\}  \geq2.
\]
Then, if $y<t_{1}$ and $y$ $\leq-1$,
\begin{align*}
\frac{1}{\mathrm{p}_{L}\left(  y\right)  }\int_{-\infty}^{y}\left(  -t\right)
\mathrm{p}_{L}\left(  t\right)  \mathrm{d}t  &  \leq\frac{b_{1}}%
{a_{1}e^{-\left\vert y\right\vert ^{\alpha_{1}}}}\int_{-\infty}^{y}\left(
-t\right)  e^{-\left\vert t\right\vert ^{\alpha_{1}}}\mathrm{d}t\\
&  \leq\frac{b_{1}}{a_{1}e^{-\left\vert y\right\vert ^{\alpha_{1}}}}%
\int_{-\infty}^{y}\left(  -t\right)  ^{\alpha_{1}-1}e^{-\left\vert
t\right\vert ^{\alpha_{1}}}\mathrm{d}t\\
&  =\frac{b_{1}}{a_{1}e^{-\left\vert y\right\vert ^{\alpha_{1}}}}%
\frac{e^{-\left\vert y\right\vert ^{\alpha_{1}}}}{\alpha_{1}-1}=\frac{b_{1}%
}{a_{1}}\frac{1}{\alpha_{1}-1}.
\end{align*}
Hence, if $\left\vert y\right\vert $ is large enough enough, we have
\[
\frac{1}{\mathrm{p}_{L}\left(  y\right)  }\int_{-\infty}^{y}\left(  -t\right)
\mathrm{p}_{L}\left(  t\right)  \mathrm{d}t<const.
\]
The same is true for $y>t_{2}$-case, and thus (\ref{exp-tail}) is another
sufficient condition for (\ref{w-variance}).

\subsection{Simulation of Gaussian shift by an arbitrary IID sequence}

\label{subsec:g-shift-sim-by-prob}

Suppose $\left\{  q^{n},\delta^{\prime n}\right\}  =\left\{  q^{\otimes
n},\delta^{\prime\left(  n\right)  }\right\}  $, where $J_{q}\left(
\delta^{\prime}\right)  =J$ , is given. Suppose also that $L^{\prime}%
:=\frac{\delta^{\prime}}{q}$ has density with respect to Lebesgue measure, and
satisfies (\ref{w-variance}). Then, by Theorem\thinspace\ref{th:cont-prob-sim}%
, we can compose asymptotic tangent simulation of
\[
\left\{  p^{\otimes n},\delta^{\left(  n\right)  }\right\}  :=\left\{
\mathrm{N}\left(  0,1\right)  ,\delta\mathrm{N}\left(  0,1\right)  \right\}
^{\otimes nJ}\equiv\left\{  \mathrm{N}\left(  0,1\right)  ,\sqrt{nJ}%
\delta\mathrm{N}\left(  0,1\right)  \right\}
\]
with $\left\{  q^{n},\delta^{\prime n}\right\}  =\left\{  q^{\otimes n}%
,\delta^{\prime\left(  n\right)  }\right\}  $.

Meanwhile, instead, suppose $\left\vert L^{\prime}\right\vert \leq const.$
with probability 1. Then, by a given \ Let $X^{i}\sim q$, and $Y^{i}%
\sim\mathrm{N}\left(  0,1\right)  $,
\[
\frac{1}{\sqrt{n}}\tilde{L}^{n}=\frac{1}{\sqrt{n}}\left(  \Lambda^{n}\right)
^{\ast}\left(  L^{\prime\left(  n\right)  }\right)  :=\frac{1}{\sqrt{n\left(
J+\varepsilon^{2}\right)  }}\sum_{i=1}^{n}\left(  L^{\prime}\left(
X^{i}\right)  +\varepsilon Y^{i}\right)  .
\]
Then \ $L^{\prime}\left(  X^{i}\right)  +\varepsilon Y^{i}$ has density with
respect to Lebesgue measure, and satisfies (\ref{exp-tail}). Since Fisher
information of $\mathrm{p}_{\tilde{L}^{n}}$ equals
\[
\frac{n}{J^{-1}+\varepsilon^{2}}=\frac{nJ}{1+\varepsilon^{2}J}=n\left(
J-f\left(  \varepsilon\right)  \right)  \,\,\,(\,\lim_{\varepsilon
\rightarrow0}f\left(  \varepsilon\right)  =0\,),
\]
by Theorem\thinspace\ref{th:cont-prob-sim}, one can compose an asymptotic
tangent symulation of
\[
\left\{  \mathrm{N}\left(  0,1\right)  ,\sqrt{n\left(  1-f\left(
\varepsilon\right)  \right)  J}\,\delta\mathrm{N}\left(  0,1\right)  \right\}
\equiv\left\{  \mathrm{N}\left(  0,1\right)  ,\delta\mathrm{N}\left(
0,1\right)  \right\}  ^{\otimes n\left(  1-f\left(  \varepsilon\right)
\right)  }%
\]
by $\left\{  \mathrm{p}_{\tilde{L}^{n}}\left(  l\right)  ,l\mathrm{p}%
_{\tilde{L}^{n}}\left(  l\right)  \right\}  $. Since $\left\{  \Lambda
^{n},\,\mathrm{p}_{L^{\prime\left(  n\right)  }}\left(  l\right)
,l\mathrm{p}_{L^{\prime\left(  n\right)  }}\left(  l\right)  \right\}  $ is an
asymptotic tangent symulation of\ $\left\{  \mathrm{p}_{\tilde{L}^{n}}\left(
l\right)  ,l\mathrm{p}_{\tilde{L}^{n}}\left(  l\right)  \right\}  $, by
Proposition\thinspace\ref{prop:sufficient} and Proposition\thinspace
\ref{prop:sym-chain}, one can compose an asymptotic tangent simulation of
$\left\{  \mathrm{N}\left(  0,1\right)  ,\delta\mathrm{N}\left(  0,1\right)
\right\}  ^{\otimes n\left(  1-f\left(  \varepsilon\right)  \right)  }$ with
$\left\{  q^{n},\delta^{\prime n}\right\}  =\left\{  q^{\otimes n}%
,\delta^{\prime\left(  n\right)  }\right\}  $.

\subsection{Uniqueness theorem}

\begin{theorem}
\label{th:uniquenss}Suppose $g$ satisfies (M0), (A0), (C0), and (N0). Suppose
also either (a): $\left\{  p,\delta\right\}  $ is defined over a finite set,
or (b): the probability density $\mathrm{p}_{L}$ of $L$ with respect to
Lebesgue measure exists and satisfies (\ref{w-variance}). Then, if
$\mathbb{E}_{p}\left(  L\right)  ^{4}<\infty$, $g_{p}\left(  \delta\right)  $
equals $J=J_{p}\left(  \delta\right)  $. \ 
\end{theorem}

\begin{proof}
Let $\left\{  q^{n},\delta^{\prime n}\right\}  :=\left\{  \mathrm{N}\left(
0,1\right)  ,\delta\mathrm{N}\left(  0,1\right)  \right\}  ^{\otimes n\left(
J+\varepsilon\right)  }=\left\{  \mathrm{N}\left(  0,1\right)  ,\sqrt{n\left(
J+\varepsilon\right)  }\delta\mathrm{N}\left(  0,1\right)  \right\}  $
($\varepsilon>0$). Then by Proposition\thinspace\ref{prop:g-shift},
\[
g_{q^{n}}\left(  \delta^{\prime n}\right)  =n\left(  J+\varepsilon\right)  .
\]
Due to Theorem\thinspace\ref{th:finite-sim-prob} and Theorem\thinspace
\ref{th:cont-prob-sim}, there is $\Lambda^{n}$ with
(\ref{simulation-tangent-prob-1}) and (\ref{simulation-tangent-prob-2}).
Therefore, by (C0) and (M0),
\begin{align*}
0  &  \leq\varliminf_{n\rightarrow\infty}\frac{1}{n}\left(  g_{\Lambda\left(
q^{n}\right)  }\left(  \Lambda\left(  \delta^{\prime n}\right)  \right)
-g_{p^{\otimes n}}\left(  \delta^{\left(  n\right)  }\right)  \right) \\
&  \leq\varliminf_{n\rightarrow\infty}\frac{1}{n}\left(  g_{q^{n}}\left(
\delta^{\prime n}\right)  -g_{p^{\otimes n}}\left(  \delta^{\left(  n\right)
}\right)  \right) \\
&  =J_{p}\left(  \delta\right)  +\varepsilon-\varlimsup_{n\rightarrow\infty
}\frac{1}{n}g_{p^{\otimes n}}\left(  \delta^{\left(  n\right)  }\right)  .
\end{align*}
Similarly, by the argument in Subsection\thinspace
\ref{subsec:g-shift-sim-by-prob}, we have,
\begin{align*}
0  &  \leq\varliminf_{n\rightarrow\infty}\frac{1}{n}\left(  g_{\Lambda\left(
p^{\otimes n}\right)  }\left(  \Lambda\left(  \delta^{\left(  n\right)
}\right)  \right)  -g_{\mathrm{N}\left(  0,1\right)  }\left(  ,\sqrt{n\left(
J-\varepsilon\right)  }\delta\mathrm{N}\left(  0,1\right)  \right)  \right) \\
&  \leq\varliminf_{n\rightarrow\infty}\frac{1}{n}\left(  g_{p^{\otimes n}%
}\left(  \delta^{\left(  n\right)  }\right)  -\left(  J_{p}\left(
\delta\right)  -\varepsilon\right)  \right)
\end{align*}
Therefore,%
\[
J_{p}\left(  \delta\right)  -\varepsilon\leq\varliminf_{n\rightarrow\infty
}\frac{1}{n}g_{p^{\otimes n}}\left(  \delta^{\left(  n\right)  }\right)
\leq\varlimsup_{n\rightarrow\infty}\frac{1}{n}g_{p^{\otimes n}}\left(
\delta^{\left(  n\right)  }\right)  \leq J_{p}\left(  \delta\right)
+\varepsilon.
\]
Since $\varepsilon>0$ is arbitrary, we have
\[
\lim_{n\rightarrow\infty}\frac{1}{n}g_{p^{\otimes n}}\left(  \delta^{\left(
n\right)  }\right)  =J_{p}\left(  \delta\right)  ,
\]
which, combined with (A0) implies
\[
g_{p}\left(  \delta\right)  =J_{p}\left(  \delta\right)  .
\]

We have to check $g_{p}\left(  \delta\right)  =J_{p}\left(  \delta\right)  $
satisfies (M0), (A0), (C0), and (N0). (A0) and (N0) are checked by easy
computation. (M0) is well-known. Hence, (C0) is shown in the sequel. We use
the following characterization of Fisher information\thinspace(see Chap. 9 of
\cite{AmariNagaoka}):
\[
J_{p}\left(  \delta\right)  =\max_{T}\frac{\left\vert \mathbb{E}%
_{p}LT\right\vert ^{2}}{\mathbb{E}_{p}T^{2}},
\]
where the maximum is achieved by $T=J^{-1}\cdot L$, with $J=J_{p}\left(
\delta\right)  $. Define $T^{n}:=\left(  nJ\right)  ^{-1}\cdot L^{\left(
n\right)  }$ , and
\begin{align*}
T_{a}\left(  x\right)   &  :=\left\{
\begin{array}
[c]{cc}%
T\left(  x\right)  , & \left(  \left\vert T\left(  x\right)  \right\vert \leq
a\right)  ,\\
0, & \left(  \left\vert T\left(  x\right)  \right\vert >a\right)  ,
\end{array}
\right. \\
T_{a}^{n}\left(  x\right)   &  :=\left\{
\begin{array}
[c]{cc}%
T^{n}\left(  x^{n}\right)  , & \left(  T^{n}\left(  x^{n}\right)  \leq
a\right)  ,\\
0, & \left(  T^{n}\left(  x^{n}\right)  >a\right)  .
\end{array}
\right.
\end{align*}
Observe
\[
\frac{1}{n}J_{q^{n}}\left(  \delta^{\prime n}\right)  \geq\frac{\left\vert
\frac{1}{\sqrt{n}}\mathbb{E}_{q^{n}}L^{\prime n}T_{a}^{n}\right\vert ^{2}%
}{\mathbb{E}_{q^{n}}\left(  T_{a}^{n}\right)  ^{2}}=\frac{\left\vert \frac
{1}{\sqrt{n}}\mathbb{E}_{p^{\otimes n}}L^{\left(  n\right)  }T_{a}%
^{n}\right\vert ^{2}}{\mathbb{E}_{p^{\otimes n}}\left(  T_{a}^{n}\right)
^{2}}+o\left(  1\right)  ,
\]
where the last identity is due to $\left\Vert q^{n}-p^{\otimes n}\right\Vert
_{1}\rightarrow0$ and $\frac{1}{\sqrt{n}}\left\Vert \delta^{\prime n}%
-\delta^{\left(  n\right)  }\right\Vert _{1}\rightarrow0$. Observe also%
\begin{align*}
\left\vert \mathbb{E}_{p^{\otimes n}}\left(  T_{a}^{n}\right)  ^{2}%
-\mathbb{E}_{p^{\otimes n}}\left(  T^{n}\right)  ^{2}\right\vert  &
=\left\vert \mathbb{E}_{p^{\otimes n}}\left(  T^{n}\right)  ^{2}\chi_{t\geq
a}\left(  T^{n}\right)  \right\vert \\
&  =\left\vert \frac{1}{J^{2}}\mathbb{E}_{p^{\otimes n}}\left(  \frac{1}%
{\sqrt{n}}L^{\left(  n\right)  }\right)  ^{2}\chi_{t\geq\sqrt{n}Ja}\left(
\frac{1}{\sqrt{n}}L^{\left(  n\right)  }\right)  \right\vert \\
&  \leq\frac{1}{J^{2}}\frac{1}{n\left(  Ja\right)  ^{2}}\mathbb{E}_{p^{\otimes
n}}\left(  \frac{1}{\sqrt{n}}L^{\left(  n\right)  }\right)  ^{4}\\
&  =\frac{1}{J^{2}}\frac{1}{n\left(  Ja\right)  ^{2}}\left(  \frac{n-1}%
{n}J+\frac{1}{n}\mathbb{E}_{p}\left(  L\right)  ^{4}\right)  =o\left(
1\right)  .
\end{align*}
Similarly,
\begin{align*}
\left\vert \frac{1}{\sqrt{n}}\mathbb{E}_{p^{\otimes n}}L^{\left(  n\right)
}T_{a}^{n}-\frac{1}{\sqrt{n}}\mathbb{E}_{p^{\otimes n}}L^{\left(  n\right)
}T^{n}\right\vert  &  =\left\vert \frac{1}{\sqrt{n}}\mathbb{E}_{p^{\otimes n}%
}L^{\left(  n\right)  }T^{n}\chi_{t\geq a}\left(  T^{n}\right)  \right\vert \\
&  =\left\vert \frac{1}{\sqrt{n}J}\mathbb{E}_{p^{\otimes n}}\left(  \frac
{1}{\sqrt{n}}L^{\left(  n\right)  }\right)  ^{2}\chi_{t\geq\sqrt{n}Ja}\left(
\frac{1}{\sqrt{n}}L^{\left(  n\right)  }\right)  \right\vert =o\left(
1\right)  .
\end{align*}
Therefore,
\begin{align*}
\frac{1}{n}J_{q^{n}}\left(  \delta^{\prime n}\right)   &  \geq\frac{\left\vert
\frac{1}{\sqrt{n}}\mathbb{E}_{p^{\otimes n}}L^{\left(  n\right)  }%
T^{n}\right\vert ^{2}}{\mathbb{E}_{p^{\otimes n}}\left(  T^{n}\right)  ^{2}%
}+o\left(  1\right)  =\frac{1}{n}\frac{\left\vert \mathbb{E}_{p^{\otimes n}%
}L^{\left(  n\right)  }T^{n}\right\vert ^{2}}{\mathbb{E}_{p^{\otimes n}%
}\left(  T^{n}\right)  ^{2}}+o\left(  1\right) \\
&  =J_{p}\left(  \delta\right)  +o\left(  1\right)  ,
\end{align*}
which is (C0).
\end{proof}

\subsection{On asymptotic continuity}

If for any $\left\{  q^{n},\delta^{\prime n}\right\}  $ with $\left\Vert
q^{n}-p^{\otimes n}\right\Vert _{1}\rightarrow0$ and $\frac{1}{\sqrt{n}%
}\left\Vert \delta^{^{\prime}n}-\delta^{\left(  n\right)  }\right\Vert
_{1}\rightarrow0$,
\[
\varliminf_{n\rightarrow\infty}\frac{1}{n}\left\vert g_{q^{n}}\left(
\delta^{\prime n}\right)  -g_{p^{\otimes n}}\left(  \delta^{\left(  n\right)
}\right)  \right\vert =0
\]
holds, we say $g$ is \textit{asymptotically continuous} at $\left\{
p^{\otimes n},\delta^{\left(  n\right)  }\right\}  $. Analogous conditions are
used in study of entanglement measures etc. In our case, Fisher information
satisfies `$\geq$', or weak asymptotic continuity, as stated in
Theorem\thinspace\ref{th:uniquenss}. However, the other side of inequality,
and thus asymptotic continuity, is false. Let $p:=\mathrm{Bin}\left(
1,t\right)  $, and
\begin{align*}
q^{n}\left(  x^{n}\right)   &  :=\left\{
\begin{array}
[c]{cc}%
\frac{t}{2}^{n}, & \left(  x^{n}=0^{n}\right) \\
\left(  1-t\right)  ^{n}+\frac{t}{2}^{n}, & \left(  x^{n}=1^{n}\right) \\
p^{\otimes n}\left(  x^{n}\right)  , & \text{otherwise}%
\end{array}
\right. \\
\delta^{^{\prime}n}\left(  x^{n}\right)   &  :=\,\delta^{\left(  n\right)
}\left(  x^{n}\right)  ,\\
\delta\left(  0\right)   &  =-\delta\left(  1\right)  =1>0,
\end{align*}
then we have $\left\Vert q^{n}-p^{\otimes n}\right\Vert _{1}=\frac{t}{2}%
^{n}+\frac{t}{2}^{n}\rightarrow0$ , $\frac{1}{\sqrt{n}}\left\Vert
\delta^{^{\prime}n}-\delta^{\left(  n\right)  }\right\Vert _{1}=0$, and%

\begin{align*}
\frac{1}{n}\left\vert J_{p^{\otimes n}}\left(  \delta^{\left(  n\right)
}\right)  -J_{q^{n}}\left(  \delta^{\prime n}\right)  \right\vert  &
=\frac{1}{n}\left\vert \left(  \frac{1}{t^{n}}-\frac{2}{t^{n}}\right)
n^{2}+\left(  \frac{1}{\left(  1-t\right)  ^{n}}-\frac{1}{\left(  1-t\right)
^{n}+\frac{t}{2}^{n}}\right)  \left(  -n\right)  ^{2}\right\vert \\
&  =\frac{1}{n}\cdot n^{2}\left\vert \frac{-1}{t^{n}}+\frac{\frac{t}{2}^{n}%
}{\left(  1-t\right)  ^{n}\left\{  \left(  1-t\right)  ^{n}+\frac{t}{2}%
^{n}\right\}  }\right\vert \rightarrow\infty.
\end{align*}

\section{Classical Channels: Non-asymptotic theory}

\subsection{Axioms}

Other than being square of a norm, $G_{\Phi}\left(  \Delta\right)  $ should satisfy:

\begin{description}
\item[(M1)] \textit{(monotonicity 1)} $G_{\Phi}\left(  \Delta\right)  \geq
G_{\Phi\circ\Psi}\left(  \Delta\circ\Psi\right)  $

\item[(M2)] \textit{(monotonicity 2)} $G_{\Phi}\left(  \Delta\right)  \geq
G_{\Psi\circ\Phi}\left(  \Psi\circ\Delta\right)  $

\item[(E)] $G_{\Phi\otimes\mathbf{I}}\left(  \Delta\otimes\mathbf{I}\right)
=G_{\Phi}\left(  \Delta\right)  $

\item[(N)] $G_{p}\left(  \delta\right)  =J_{p}\left(  \delta\right)  $
\end{description}

\subsection{Simulation of channel families}

Suppose we have to fabricate a channel $\Phi_{\theta}$, which is drawn from a
family $\left\{  \Phi_{\theta}\right\}  $, without knowing the value of
$\theta$ but with a probability distribution $q_{\theta\text{ }}$or a channel
$\Psi_{\theta}$, drawn from a family $\left\{  q_{\theta}\right\}  $ or
$\left\{  \Psi_{\theta}\right\}  $. More specifically, we need a channel
$\Lambda$ with
\begin{equation}
\Phi_{\theta}=\Lambda\circ\left(  \mathbf{I}\otimes q_{\theta}\right)  ,
\label{simulation-1}%
\end{equation}
Here, note that $\Lambda$ should not vary with the parameter $\theta$. Giving
the value of $\theta$ with infinite precision corresponds to the case where
$q_{\theta}$ is delta distribution centered at $\theta$.

Differentiating the both ends of (\ref{simulation-1}) and letting
$\Phi_{\theta}=\Phi$ and $q_{\theta}=q$, we obtain%

\begin{equation}
\Delta=\Lambda\circ\left(  \mathbf{I}\otimes\delta^{\prime}\right)  ,\text{ }
\label{simulation-tangent-1}%
\end{equation}
where $\Delta\in\mathcal{T}_{\Phi}\left(  \mathcal{C}\right)  $ and
$\delta^{\prime}\in\mathcal{T}_{q}\left(  \mathcal{P}^{\prime}\right)  $.

In the manuscript, we consider \textit{tangent simulation}, or the operations
satisfying (\ref{simulation-1}) and (\ref{simulation-tangent-1}), at the point
$\Phi_{\theta}=\Phi$ only. Note that simulation of $\left\{  \Phi
,\Delta\right\}  $ is equivalent to the one of the channel family $\left\{
\Phi_{\theta+t}=\Phi+t\Delta\right\}  _{t}$.

\subsection{Relation between $J$ and $G$}

In this section, we review quickly the properties of norms with (M1), (M2),
(E), and (N). For the proof, see\thinspace\ \cite{Matsumoto:2010-1}.

\begin{theorem}
\label{th:G>Gmin}Suppose (M1) and (N) hold. Then, $\ $%
\[
G_{\Phi}\left(  \Delta\right)  \geq G_{\Phi}^{\min}\left(  \Delta\right)
:=\sup_{p\in\mathcal{P}_{\mathrm{in}}}J_{\Phi\left(  p\right)  }\left(
\Delta\left(  p\right)  \right)  =\sup_{x\in\mathcal{\Omega}_{\mathrm{in}}%
}J_{\Phi\left(  \cdot|x\right)  }\left(  \Delta\left(  \cdot|x\right)
\right)  .
\]
Trivially, $G_{\Phi}^{\min}\left(  \Delta\right)  $ satisfies (M1), (M2), (E),
and (N).
\end{theorem}

\begin{theorem}
\label{th:G<Gmax}Suppose (M2), (E) and (N) hold. Then
\[
G_{\Phi}\left(  \Delta\right)  \leq G_{\Phi}^{\max}\left(  \Delta\right)
:=\inf_{\Lambda,q,\delta}\left\{  J_{q}\left(  \delta\right)  ;\,\,\Lambda
\circ\left(  \mathbf{I}\otimes q\right)  =\Phi,\,\Lambda\circ\left(
\mathbf{I}\otimes\delta\right)  =\Delta\text{ }\right\}  .
\]
Also, $G_{\Phi}^{\max}\left(  \Delta\right)  $ satisfies (M1), (M2), (E), and (N).
\end{theorem}

Obviously, $G_{\Phi}^{\min}\left(  \Delta\right)  $ and $G_{\Phi}^{\max
}\left(  \Delta\right)  $ are not induced from any metric, i.e., they cannot
be written as $S\left(  \Delta,\Delta\right)  $, where $S$ is a positive real
bilinear form. Indeed, we can show the following :

\begin{theorem}
\label{th:no-inner-product}Suppose (M1), (N), and (E) hold. Then, $G_{\Phi
}\left(  \Delta\right)  $ cannot written as $S_{\Phi}\left(  \Delta
,\Delta\right)  $, where $S$ is a positive bilinear form.
\end{theorem}

\ \ 

\section{ Classical Channels: Asymptotic Theory}

\subsection{Asymptotic Theory: additional axioms}

\begin{description}
\item[(A)] \textit{(asymptotic weak additivity)} $\lim_{n\rightarrow\infty
}\frac{1}{n}G_{\Phi^{\otimes n}}\left(  \Delta^{\left(  n\right)  }\right)
=G_{\Phi}\left(  \Delta\right)  $

\item[(C)] \textit{(weak asymptotic continuity)} If $\left\Vert \Phi^{n}%
-\Phi^{\otimes n}\right\Vert _{\mathrm{cb}}\rightarrow0$ and $\frac{1}%
{\sqrt{n}}\left\Vert \Delta^{n}-\Delta^{\left(  n\right)  }\right\Vert
_{\mathrm{cb}}\rightarrow0$ then
\[
\varliminf_{n\rightarrow\infty}\frac{1}{n}\left(  G_{\Phi^{n}}\left(
\Delta^{n}\right)  -G_{\Phi^{\otimes n}}\left(  \Delta^{\left(  n\right)
}\right)  \right)  \geq0.
\]

\end{description}

\subsection{Asymptotic tangent simulation:definition}

We consider asymptotic version of approximate version of (\ref{simulation-1}%
)-(\ref{simulation-tangent-2}):
\begin{equation}
\lim_{n\rightarrow\infty}\left\Vert \Phi^{\otimes n}\left(  p\right)
-\Lambda^{n}\left(  p\otimes q^{n}\right)  \right\Vert _{\mathrm{cb}%
}=0\,,\,\forall p, \label{e-simulation-2}%
\end{equation}
and%

\begin{equation}
\lim_{n\rightarrow\infty}\frac{1}{\sqrt{n}}\left\Vert \Delta^{\left(
n\right)  }\left(  p\right)  -\Lambda^{n}\left(  p\otimes\delta^{n}\right)
\right\Vert _{\mathrm{cb}}=0,\,\forall p\text{ },
\label{e-simulation-tangent-2}%
\end{equation}
with "program" $\left\{  q^{n},\,\delta^{n}\right\}  $. Here, the larger one
of $\left\Vert \Phi^{\otimes n}\left(  p\right)  -\Lambda^{n}\left(  p\otimes
q^{n}\right)  \right\Vert _{\mathrm{cb}}$ and $\frac{1}{\sqrt{n}}\left\Vert
\Delta^{\left(  n\right)  }\left(  p\right)  -\Lambda^{n}\left(
p\otimes\delta^{n}\right)  \right\Vert _{\mathrm{cb}}$ is called the
\textit{error} of the asymtotic tangent simulation.

\subsection{Finite inputs}

In this sutbsection, $\Omega_{\mathrm{in}}$ $=\left\{  1,\cdots,k\right\}  $.

\begin{theorem}
\label{th:channel-sim-finite}Suppose $\left\{  \Phi\left(  \cdot|x\right)
,\Delta\left(  \cdot|x\right)  \right\}  $ satisfies all the conditions
imposed on $\left\{  p,\delta\right\}  $ in Theorem\thinspace
\ref{th:uniquenss}. Let us define $\left\{  q_{\varepsilon}^{n},\delta
_{\varepsilon}^{n}\right\}  :=$ $\{\mathrm{N}\left(  0,1\right)
,\delta\mathrm{N}\left(  0,1\right)  \}^{\otimes n\left(  1+k\varepsilon
\right)  \left(  J+c\right)  }$ where $J=G_{\Phi}^{\min}\left(  \Delta\right)
=\max_{1\leq x\leq k}J_{\Phi\left(  \cdot|x\right)  }\left(  \Delta\left(
\cdot|x\right)  \right)  $ and $\varepsilon>0$, $c>0$ are arbitrary. Then,
there is $\Lambda^{n}$ such that
\begin{align}
\left\Vert \Phi^{\otimes n}\left(  p\right)  -\Lambda^{n}\left(  p\otimes
q_{\varepsilon}^{n}\right)  \right\Vert _{\mathrm{cb}}  &  \leq\frac
{C}{\left(  \varepsilon n\right)  ^{1/4}},\nonumber\\
\frac{1}{\sqrt{n}}\left\Vert \Delta^{\left(  n\right)  }\left(  p\right)
-\Lambda^{n}\left(  p\otimes\delta_{\varepsilon}^{n}\right)  \right\Vert
_{\mathrm{cb}}  &  \leq\frac{C}{\left(  \varepsilon n\right)  ^{1/4}},
\label{sim-channel-error}%
\end{align}
where $C$ is a function of $\left\{  \Phi\left(  y|x\right)  ,\Delta\left(
y|x\right)  ;x\in\Omega_{\mathrm{in}},y\in\Omega_{\mathrm{out}}\right\}  $.
Especially, if $\left\vert \Omega_{\mathrm{out}}\right\vert <\infty$, this
function is countinuous and bounded.
\end{theorem}

\begin{proof}
Given the input sequence $x^{n}=x_{1}\cdots x_{n}$, denote the number of $x$
in $x^{n}$ by $N_{x}$. Suppose $N_{x}\geq\varepsilon n$. Then, we use
$\{\mathrm{N}\left(  0,1\right)  ,\delta\mathrm{N}\left(  0,1\right)
\}^{\otimes N_{x}\left(  J+c\right)  }$ for simulation of $\left\{
\Phi\left(  \cdot|x\right)  ,\Delta\left(  \cdot|x\right)  \right\}  ^{\otimes
N_{x}}$. On the other hand, if $N_{x}<\varepsilon n$, we first fabricate
$\left\{  \Phi\left(  \cdot|x\right)  ,\Delta\left(  \cdot|x\right)  \right\}
^{\otimes\varepsilon n}$ using $\{\mathrm{N}\left(  0,1\right)  ,\delta
\mathrm{N}\left(  0,1\right)  \}^{\otimes n\varepsilon\left(  J+c\right)  }$,
and takes marginal. We do this for all $x=1,\cdots,k$. Since
\[
\bigotimes_{x=1}^{k}\{\mathrm{N}\left(  0,1\right)  ,\delta\mathrm{N}\left(
0,1\right)  \}^{\otimes N_{x}\left(  1+\varepsilon\right)  \left(
J_{x}+c\right)  }\equiv\{\mathrm{N}\left(  0,1\right)  ,\delta\mathrm{N}%
\left(  0,1\right)  \}^{\otimes n\left(  1+k\varepsilon\right)  \left(
J+c\right)  },
\]
by Theorem\thinspace\ref{th:finite-sim-prob} and Theorem\thinspace
\ref{th:cont-prob-sim}, we have (\ref{sim-channel-error}) and the proof is complete.
\end{proof}

\begin{theorem}
\label{channel-cont-unique}Suppose $\left\{  \Phi\left(  \cdot|x\right)
,\Delta\left(  \cdot|x\right)  \right\}  $ satisfies all the conditions
imposed on $\left\{  p,\delta\right\}  $ in Theorem\thinspace
\ref{th:uniquenss} for all $x\in\Omega_{\mathrm{in}}$. Then, if a metric $G$
satisfies (M1), (M2), (E), (A), (C), and (N0). Then,
\[
G_{\Phi}\left(  \Delta\right)  =G_{\Phi}^{\min}\left(  \Delta\right)  .
\]

\end{theorem}

\begin{proof}
Due to Theorem\thinspace\ref{th:G>Gmin}, we only have to show $G_{\Phi}\left(
\Delta\right)  \leq G_{\Phi}^{\min}\left(  \Delta\right)  $. Consider the
simulation of $\left\{  \Phi,\Delta\right\}  $ by $\left\{  q_{\varepsilon
}^{n},\delta_{\varepsilon}^{n}\right\}  $ as of Theorem\thinspace
\ref{th:channel-sim-finite}. Due to Theorem\thinspace\ref{th:uniquenss},
\[
G_{q_{\varepsilon}^{n}}\left(  \delta_{\varepsilon}^{n}\right)
=J_{q_{\varepsilon}^{n}}\left(  \delta_{\varepsilon}^{n}\right)  .
\]
Therefore, due to (\ref{sim-channel-error}), we have
\begin{align*}
&  0\underset{\text{(C)}}{\leq}\varliminf_{n\rightarrow\infty}\frac{1}%
{n}\left(  G_{\Lambda\circ\left(  \mathbf{I}\otimes q_{\varepsilon}%
^{n}\right)  }\left(  \Lambda\circ\left(  \mathbf{I}\otimes\delta
_{\varepsilon}^{n}\right)  \right)  -G_{\Phi^{\otimes n}}\left(
\Delta^{\left(  n\right)  }\right)  \right) \\
&  \underset{\text{(M)}}{\leq}\varliminf_{n\rightarrow\infty}\frac{1}%
{n}\left(  G_{\mathbf{I}\otimes q_{\varepsilon}^{n}}\left(  \mathbf{I}%
\otimes\delta_{\varepsilon}^{n}\right)  -G_{\Phi^{\otimes n}}\left(
\Delta^{\left(  n\right)  }\right)  \right) \\
&  \underset{\text{(E)}}{=}\varliminf_{n\rightarrow\infty}\frac{1}{n}\left(
G_{q_{\varepsilon}^{n}}\left(  \delta_{\varepsilon}^{n}\right)  -G_{\Phi
^{\otimes n}}\left(  \Delta^{\left(  n\right)  }\right)  \right) \\
&  =\varliminf_{n\rightarrow\infty}\frac{1}{n}\left(  J_{q_{\varepsilon}^{n}%
}\left(  \delta_{\varepsilon}^{n}\right)  -G_{\Phi^{\otimes n}}\left(
\Delta^{\left(  n\right)  }\right)  \right) \\
&  \leq\left(  1+\varepsilon k\right)  \left(  G_{\Phi}^{\min}\left(
\Delta\right)  +c\right)  -\varlimsup_{n\rightarrow\infty}\frac{1}{n}%
G_{\Phi^{\otimes n}}\left(  \Delta^{\left(  n\right)  }\right) \\
&  \underset{\text{(A)}}{=}\left(  1+\varepsilon k\right)  \left(  G_{\Phi
}^{\min}\left(  \Delta\right)  +c\right)  -G_{\Phi}\left(  \Delta\right)  .
\end{align*}
Since $\varepsilon>0$ and $c>0$ are arbitrary, we have
\[
G_{\Phi}\left(  \Delta\right)  \leq G_{\Phi}^{\min}\left(  \Delta\right)  .
\]
Finally, we show $G_{\Phi}^{\min}\left(  \Delta\right)  $ satisfies (C). Let
$x_{\ast}\in\Omega_{\mathrm{in}}$ with $J_{\Phi\left(  \cdot|x\right)
}\left(  \Delta\left(  \cdot|x_{\ast}\right)  \right)  =G_{\Phi}^{\min}\left(
\Delta\right)  $, and $x_{\ast}^{n}=x_{\ast}x_{\ast}\cdots x_{\ast}$. Then,
since $G_{\Psi^{n}}^{\min}\left(  \Delta^{\prime n}\right)  \geq J_{\Psi
^{n}\left(  \cdot|x_{\ast}^{n}\right)  }\left(  \Delta^{\prime n}\left(
\cdot|x_{\ast}^{n}\right)  \right)  $, we have
\[
\varliminf_{n\rightarrow\infty}\frac{1}{n}\left(  G_{\Psi^{n}}^{\min}\left(
\Delta^{\prime n}\right)  -G_{\Phi^{\otimes n}}^{\min}\left(  \Delta^{\left(
n\right)  }\right)  \right)  \geq\varliminf_{n\rightarrow\infty}\frac{1}%
{n}\left\{  J_{\Psi^{n}\left(  \cdot|x\right)  }\left(  \Delta^{\prime
n}\left(  \cdot|x_{\ast}^{n}\right)  \right)  -J_{\Phi\left(  \cdot|x\right)
^{\otimes n}}\left(  \Delta\left(  \cdot|x_{\ast}\right)  ^{^{\left(
n\right)  }}\right)  \right\}  .
\]
The LHS of this is non-negative due to Theorem\thinspace\ref{th:uniquenss},
since
\begin{align*}
\left\Vert \Psi^{n}\left(  \cdot|x_{\ast}^{n}\right)  -\Phi\left(
\cdot|x_{\ast}\right)  ^{\otimes n}\right\Vert _{\mathrm{1}}  &
\leq\left\Vert \Psi^{n}-\Phi^{\otimes n}\right\Vert _{\mathrm{cb}}=o\left(
1\right)  ,\\
\frac{1}{\sqrt{n}}\left\Vert \Delta^{\prime n}\left(  \cdot|x_{\ast}%
^{n}\right)  -\Delta\left(  \cdot|x_{\ast}\right)  ^{\left(  n\right)
}\right\Vert _{\mathrm{1}}  &  \leq\frac{1}{\sqrt{n}}\left\Vert \Delta^{\prime
n}-\Delta^{\left(  n\right)  }\right\Vert _{\mathrm{cb}}=o\left(  1\right)  .
\end{align*}

\end{proof}

\subsection{Continuous inputs}

In this subsection, $\Omega_{\mathrm{in}}$ is a compact set in $%
\mathbb{R}
^{d}$. \ Also, $\left\Vert x\right\Vert $ is usual 2-norm. 

\begin{theorem}
\label{th:sim-channel-cont}Suppose $\left\vert \Omega_{\mathrm{out}%
}\right\vert <\infty$ and
\[
\max\left\{  \left\Vert \Phi\left(  \cdot|x\right)  -\Phi\left(
\cdot|x^{\prime}\right)  \right\Vert _{1},\left\Vert \Delta\left(
\cdot|x\right)  -\Delta\left(  \cdot|x^{\prime}\right)  \right\Vert \right\}
\leq f\left(  \left\Vert x-x^{\prime}\right\Vert \right)
\]
holds for some $\lim_{t\rightarrow0}f\left(  t\right)  =0$. Let us define
$\left\{  q_{\varepsilon}^{n},\delta_{\varepsilon}^{n}\right\}  :=$
$\{\mathrm{N}\left(  0,1\right)  ,\delta\mathrm{N}\left(  0,1\right)
\}^{\otimes n\left(  1+k\varepsilon\right)  \left(  J+c\right)  }$ where
$J=G_{\Phi}^{\min}\left(  \Delta\right)  =\max_{1\leq x\leq k}J_{\Phi\left(
\cdot|x\right)  }\left(  \Delta\left(  \cdot|x\right)  \right)  $ and
$\varepsilon>0$, $c>0$ are arbitrary. Then, there is a family $\left\{
\Phi_{t},\Delta_{t}\right\}  _{t\geq0}$ and $\left\{  \Lambda_{t}%
^{n},q_{\varepsilon,t}^{n},\delta_{\varepsilon,t}^{^{\prime}n}\right\}
_{t\geq0}$ such that%
\begin{align}
\left\Vert \Phi_{t}^{\otimes n}\left(  p\right)  -\Lambda_{t}^{n}\left(
p\otimes q_{\varepsilon,t}^{n}\right)  \right\Vert _{\mathrm{cb}} &  \leq
\frac{C_{t}}{\sqrt{\varepsilon n}}\label{cont-sim-1}\\
\frac{1}{\sqrt{n}}\left\Vert \Delta_{t}^{\left(  n\right)  }\left(  p\right)
-\Lambda_{t}^{n}\left(  p\otimes\delta_{\varepsilon,t\text{ }}^{^{\prime}%
n}\right)  \right\Vert _{\mathrm{cb}} &  \leq\frac{C_{t}}{\left(  \varepsilon
n\right)  ^{1/4}}\,\,\,\,(\lim_{t\rightarrow0}C_{t}<\infty),\label{cont-sim-2}%
\\
\lim_{t\rightarrow0}\left\Vert \Phi_{t}-\Phi\right\Vert _{\mathrm{cb}} &
=\lim_{t\rightarrow0}\left\Vert \Delta_{t}-\Delta\right\Vert _{\mathrm{cb}%
}=0.\label{cont-sim-3}%
\end{align}

\end{theorem}

\begin{proof}
Let \ $\mathcal{A}_{t}\subset\Omega_{\mathrm{in}}=%
\mathbb{R}
^{d}$ be the totality of lattice points such that \ $\min_{x,y\in
\mathcal{A}_{t}}\left\Vert x-y\right\Vert =t$. Define
\[
\Phi_{t}\left(  \cdot|x\right)  :=\Phi\left(  \cdot|y\right)  ,\quad\Delta
_{t}\left(  \cdot|x\right)  :=\Delta\left(  \cdot|y\right)  ,
\]
where $y$ is the closest point in $\mathcal{A}_{t}$ to $x$. By assumption,
$\left\{  \Phi_{t},\Delta_{t}\right\}  $ satisfies (\ref{cont-sim-3}). By
Theorem\thinspace\ref{th:channel-sim-finite},  we can compose $\Lambda_{t}%
^{n}$ \ with (\ref{cont-sim-1}) and (\ref{cont-sim-2}).
\end{proof}

\begin{description}
\item[(C2)] If $\lim_{t\rightarrow0}\left\Vert \Phi_{t}-\Phi\right\Vert
_{\mathrm{cb}}=\lim_{t\rightarrow0}\left\Vert \Delta_{t}-\Delta\right\Vert
_{\mathrm{cb}}=0$, $\lim_{t\rightarrow0}G_{\Phi_{t}}\left(  \Delta_{t}\right)
=G_{\Phi}\left(  \Delta\right)  .$
\end{description}

\begin{theorem}
Suppose $\left\{  \Phi,\Delta\right\}  $ satisfies all the asumptions of
Theorem\thinspace\ref{th:sim-channel-cont}. Then, if $G$ satisfies (M1), (M2),
(E), (A), (C), (N0), and (C2), $\ $
\[
G_{\Phi}\left(  \Delta\right)  =G_{\Phi}^{\min}\left(  \Delta\right)  .
\]

\end{theorem}

\begin{proof}
Again, we only have to show $G_{\Phi}\left(  \Delta\right)  \leq G_{\Phi
}^{\min}\left(  \Delta\right)  $ . By Theorem \ref{th:sim-channel-cont},
$G_{\Phi_{t}}\left(  \Delta_{t}\right)  =G_{\Phi_{t}}^{\min}\left(  \Delta
_{t}\right)  $. Therefore, due to (C2),
\[
G_{\Phi}\left(  \Delta\right)  =\lim_{t\rightarrow0}G_{\Phi_{t}}\left(
\Delta_{t}\right)  =\lim_{t\rightarrow0}G_{\Phi_{t}}^{\min}\left(  \Delta
_{t}\right)  .
\]
On the other hand, by construction of $\left\{  \Phi_{t},\Delta_{t}\right\}
$,
\[
G_{\Phi_{t}}^{\min}\left(  \Delta_{t}\right)  \leq\sup_{x}J_{\Phi\left(
\cdot|x\right)  }\left(  \Delta\left(  \cdot|x\right)  \right)  =G_{\Phi
}^{\min}\left(  \Delta\right)  .
\]
Hence, we have the assertion.
\end{proof}

\subsection{Quantum states as a classical channel}

A quantum state can be viewed as a channel which takes a measurement as an
input, and outputs measurement result. Hence, if we restrict the measurements
to separable measurements, the asymptotic theory discussed in this paper is
applicable to quantum states also, proving the uniqueness of the metric. On
the other hand, there are variety of monotone metrics, and lower asymptotic
continuity is proven for some of them, e.g., SLD and RLD metric. This
appearent contradiction can be circumvented by recalling that the theory of
this paper is not applicable to the case of collective measurement.

\bigskip
\end{document}